\documentclass[11pt]{article}

\PassOptionsToPackage{table,usenames,dvipsnames}{xcolor}
\usepackage[utf8]{inputenc} 
\usepackage[T1]{fontenc}    
\usepackage{hyperref}       
\usepackage{url}            
\usepackage{booktabs}       
\usepackage{amsfonts}       
\usepackage{nicefrac}       
\usepackage{fullpage}

\newif\ifarxiv\arxivtrue

\newif\ifnotarxiv
\ifarxiv \notarxivfalse \else \notarxivtrue \fi

\ifarxiv
\newcommand{\arxiv}[1]{#1}            							
\else
\newcommand{\arxiv}[1]{}
\fi

\ifarxiv
\else
\renewcommand{\subsection}[1]{\paragraph{{#1}.}}
\fi

\ifarxiv
\else
\renewcommand{\subsubsection}[1]{\paragraph{{#1}.}}
\fi

\usepackage{makeidx}
\usepackage{euscript}
\usepackage{amsmath,amssymb,amsfonts}
\usepackage{dsfont}
\usepackage{latexsym}
\usepackage{graphicx}
\usepackage{fancybox}
\usepackage{tcolorbox}
\usepackage{hyperref}
\usepackage{paralist}
\usepackage{wrapfig}
\usepackage{tikz}
\usetikzlibrary{decorations.pathreplacing}
\usepackage{setspace}
\usepackage[noend]{algpseudocode}
\usepackage[noend,ruled,vlined]{algorithm2e}
\usepackage[framemethod=tikz]{mdframed}
\usepackage{xspace}
\usepackage{pgfplots}
\usepackage{framed}
\usepackage{float}
\usepackage{subfig}
\usepackage{units}

\pgfplotsset{compat=1.5}

\newenvironment{proof}{\noindent{\bf Proof : \ }}{\hfill$\Box$\par\medskip}

\newtheorem{theorem}{Theorem}[section]
\makeatletter
\renewcommand*{\@opargbegintheorem}[3]{\trivlist
      \item[\hskip \labelsep{\bfseries #1\ #2}] \textbf{(#3)}\ \itshape}
\makeatother
\newtheorem{corollary}[theorem]{Corollary}
\newtheorem{lemma}[theorem]{Lemma}
\newtheorem{proposition}[theorem]{Proposition}
\newtheorem{definition}[theorem]{Definition}

\newtheorem{example}[theorem]{Example}
\newtheorem{fact}[theorem]{Fact}

\newenvironment{proofof}[1]{\begin{trivlist} \item {\bf Proof
#1:~~}}
  {\qed\end{trivlist}}

\newcommand{\namedref}[2]{\hyperref[#2]{#1~\ref*{#2}}}

\newcommand{\best}{\ensuremath{o_1}}
\newcommand{\gval}{\ensuremath{g}}
\newcommand{\smallhval}{\ensuremath{g_+}}
\newcommand{\bigHval}{\ensuremath{g_1}}
\newcommand{\tval}{\ensuremath{t}}
\newcommand{\smallrval}{\ensuremath{t_+}}
\newcommand{\bigRval}{\ensuremath{t_1}}
\newcommand{\margain}[2]{f\left( #1\, \middle| \, #2 \right)}
\newcommand{\marden}[2]{\rho\left( #1\, \middle| \, #2 \right)}

\renewcommand{\O}[1]{\ensuremath{\mathcal{O}\left(#1\right)}}
\newcommand{\Oin}[1]{\ensuremath{\mathcal{O}(#1)}}
\newcommand{\eps}{\epsilon}
\newcommand{\LB}{\mathsf{LB}}
\newcommand{\clast}{c^*}
\newcommand{\cbest}{c\left(\best\right)}
\newcommand{\calG}{\mathcal{G}}
\newcommand{\calT}{\mathcal{T}}
\newcommand{\opt}{\mathsf{OPT}}

\newcommand{\mdef}[1]{{\ensuremath{#1}}\xspace}  

\DeclareMathOperator*{\argmax}{argmax}

\DeclareMathOperator*{\eexp}{exp}

\newcommand{\superscript}[1]{\ensuremath{^{\mbox{\tiny{\textit{#1}}}}}\xspace}
\def \th {\superscript{th}}     

\newcommand{\set}[1]{\mdef{\left\{#1\right\}}}                        

\newcommand{\ignore}[1]{}

\newif\ifnotes\notestrue 
\ifnotes
\newcommand{\samson}[1]{\textcolor{purple}{{\bf (Samson:} {#1}{\bf ) }} \marginpar{\tiny\bf
             \begin{minipage}[t]{0.5in}
               \raggedright S:
            \end{minipage}}}            							
\else
\newcommand{\samson}[1]{}
\fi

\hypersetup{
     colorlinks   = true,
     citecolor    = blue,
	 linkcolor		= red
}

\makeatletter
\renewcommand*{\@fnsymbol}[1]{\textcolor{blue}{\ensuremath{\ifcase#1\or *\or \dagger\or \ddagger\or
 \mathsection\or \triangledown\or \mathparagraph\or \|\or **\or \dagger\dagger
   \or \ddagger\ddagger \else\@ctrerr\fi}}}
\makeatother

\providecommand{\email}[1]{\href{mailto:#1}{\nolinkurl{#1}\xspace}}

\usepackage{stmaryrd}

\newcommand{\COMMENTED}[1]{{}}

\def \Rp       {\mdef{\mathbb{R}^{+}}}                   

\renewcommand{\exp}[1]{\ensuremath{e^{#1}}}

\newcommand{\gmax}{\textsc{Greedy+Max}\xspace}
\newcommand{\gormax}{\textsc{GreedyOrMax}\xspace}
\newcommand{\smax}{\textsc{Sieve+Max}\xspace}
\newcommand{\dmax}{\textsc{Distributed Sieve+Max}\xspace}
\newcommand{\greedy}{\textsc{Greedy}\xspace}
\newcommand{\sieve}{\textsc{Sieve}\xspace}
\newcommand{\sormax}{\textsc{SieveOrMax}\xspace}
\newcommand{\partenum}{\textsc{PartialEnum+Greedy}\xspace}
\newcommand{\branching}{\textsc{BranchingMRT}\xspace}

\newcommand{\machines}{{m}}
\newcommand{\minnk}{\tilde K}

\newcommand{\half}{\nicefrac{1}{2}}

\begin{document}
\title{\textsc{``Bring Your Own Greedy''+Max}: Near-Optimal $\half$-Approximations for Submodular Knapsack}
\author{Dmitrii Avdiukhin\thanks{Indiana University, Bloomington. \email{davdyukh@iu.edu}.}
\and
Grigory Yaroslavtsev\thanks{Indiana University, Bloomington \& The Alan Turing Institute. \email{gyarosla@iu.edu}.}
\and
Samson Zhou\thanks{Carnegie Mellon University \& Indiana University, Bloomington. \email{samsonzhou@gmail.com}.}
}
\maketitle

\begin{abstract}
The problem of selecting a small-size representative summary of a large dataset is a cornerstone of machine learning, optimization and data science. 
Motivated by applications to recommendation systems and other scenarios with query-limited access to vast amounts of data, we propose a new rigorous algorithmic framework for a standard formulation of this problem as a submodular maximization subject to a linear (knapsack) constraint. 
Our framework is based on augmenting all partial \greedy solutions with the best additional item. 
It can be instantiated with negligible overhead in any model of computation, which allows the classic \greedy algorithm and its variants to be implemented. 
We give such instantiations in the offline (\gmax), multi-pass streaming (\smax) and distributed (\dmax) settings. 
Our algorithms give ($\half-\eps$)-approximation with most other key parameters of interest being near-optimal. 
Our analysis is based on a new set of first-order linear differential inequalities and their robust approximate versions. 
Experiments on typical datasets (movie recommendations, influence maximization) confirm scalability and high quality of solutions obtained via our framework. 
Instance-specific approximations are typically in the 0.6-0.7 range and frequently beat even the $(1-1/e) \approx 0.63$ worst-case barrier for polynomial-time algorithms.
\end{abstract}

\ignore{
\listoftodos
\arxiv{
\newpage

\tableofcontents
\newpage
}
}

\section{Introduction}
A fundamental problem in many large-scale machine learning, data science and optimization tasks is finding a small representative subset of a big dataset. 
This problem arises from applications in recommendation systems~\cite{LeskovecKGFVG07,ElAriniG11,BogunovicMSC17,MitrovicBNTC17,YuXC18,AvdiukhinMYZ19}, exemplar-based clustering~\cite{GomesK10}, facility location~\cite{LindgrenWD16}, image processing~\cite{IyerB19}, viral marketing~\cite{HartlineMS2008}, principal component analysis~\cite{KhannaGPK15}, and document summarization~\cite{LinB11,WeiLKB13,SiposSSJ12}
and can often be formulated as constrained monotone submodular optimization under various constraints such as cardinality~\cite{BadanidiyuryMKK14,BateniEM18,KazemiMZLK19}, knapsack~\cite{HuangKY17}, matchings~\cite{ChakrabartiK14}, and matroids~\cite{CalinescuCPV11, AnariHNP19} due to restrictions demanded by space, budget, diversity, fairness or privacy. 
As a result, constrained submodular optimization has been recently and extensively studied in various computational models, including centralized~\cite{NemhauserWF78}, distributed~\cite{MirzasoleimanKSK13,KumarMVV15,BarbosaENW15,MirrokniZ15,MirzasoleimanZK16,BarbosaENW16,LiuV19}, streaming~\cite{BadanidiyuryMKK14,BuchbinderFS15,NorouziFardTMZMS18,AgarwalSS18,KazemiMZLK19}, and adaptive~\cite{GolovinK11,BalkanskiS18,BalkanskiRS19,FahrbachMZ19,EneN19,ChekuriQ19} among others. 

In this paper we focus on monotone submodular maximization \emph{under a knapsack constraint}, which captures the scenario when the representative subset should have a small cost or size. 
While a number of algorithmic techniques exist for this problem, there are few that robustly scale to large data and can be easily implemented in various computing frameworks. 
This is in contrast with a simpler \emph{cardinality-constrained} version in which only the number of elements is restricted. 
In this setting the celebrated \greedy algorithm of~\cite{NemhauserWF78} enjoys both an optimal approximation ratio and a simplicity that allows easy adaptation in various environments. 
For knapsack constraints, such a simple and universal algorithm is unlikely. 
In particular, \greedy does not give any approximation guarantee.

We develop a framework that augments solutions constructed by \greedy and its variations and gives almost $\nicefrac{1}{2}$-approximations\footnote{Algorithm gives an $\alpha$-approximation if it outputs $S$ such that $f(S)\ge\alpha f(\opt)$\arxiv{, where $\opt$ is optimum solution}.} in various computational models. 
For example, in the multi-pass streaming setting we achieve optimal space and almost optimal number of queries and running time. 
We believe that our framework is robust to the choice of the computational model as it can be implemented with essentially the same complexity as that of running \greedy and its variants.

\subsection{Preliminaries and our contributions}\label{sec:results}
A set function $f \colon 2^U \to \mathbb R$ is \emph{submodular} if for every $S \subseteq T \subseteq U$ and $e \in U$ it holds that $f(e \cup T) - f(T)\le f(e \cup S)-f(S)$.  
Moreover, $f$ is \textit{monotone} if for every $S \subseteq T \subseteq U$ it holds that $f(T) \ge f(S)$. 
\arxiv{Intuitively, elements in the universe contribute non-negative utility, but their resulting gain is diminishing as the size of the set increases.} 
In the monotone submodular maximization problem subject to a knapsack constraint, each item $e$ has cost $c(e)$.
Given a parameter $K>0$, the task is to maximize a non-negative monotone submodular function $f(S)$ under the constraint $c(S):=\sum_{e\in S}c(e)\le K$. 
Without loss of generality, we assume that $\min_{e\in S}c(e) \ge 1$, which can be achieved by rescaling the costs and taking all items with cost $0$. 
Then $\minnk = \min(n, K)$ is an upper bound on the number of elements in any feasible solution.

Any algorithm for submodular maximization requires \emph{query access} to $f$. 
As query access can be expensive, the number of queries is typically considered one of the performance metrics. 
Furthermore, in some critical applications of submodular optimization such as recommendation systems, another constraint often arises from the fact that only queries to feasible sets are allowed (e.g. when click-through rates can only be collected for sets of ads which can be displayed to the users). 
Practical algorithms for submodular optimization hence typically only make such queries, an assumption commonly used in the literature (see e.g.~\cite{NorouziFardTMZMS18}). 
For \emph{any} algorithm that only makes queries on feasible sets, it is easy to show that $\Omega(n^2)$ queries are required to go beyond $\nicefrac12$-approximation under various assumptions on $f$ (Theorem~\ref{thm:query-lb}).  
Hence it is natural to ask whether we can get a $\nicefrac12$-approximation, while keeping other performance metrics of interest nearly optimal and hence not compromising on practicality. 
We answer this question positively.

We first state the following simplified result in the most basic \emph{offline model} (i.e. when an algorithm can access any element at any time) to illustrate the main ideas and then improve parameters in our other results. 
In this model, we are given an integer knapsack capacity $K\in\mathbb{Z}^+$ and a set $E$ of elements $e_1,\ldots, e_n$ from a finite universe $U$.\footnote{W.l.o.g. for all $e$ we have $1\le c(e)\le K$  as one can rescale the capacity and costs and filter out all items with cost more than $K$ (in all our results this means replacing $K$ with the aspect ratio $K/\min_{e \in E} c(e)$).}

\begin{theorem}[Offline \gmax]\label{thm:gmax} \\
Let $\minnk = \min(n, K)$. There exists an offline algorithm \gmax (Algorithm~\ref{alg:gmax}) that gives a $\half$-approximation for the submodular maximization problem under a knapsack constraint with query complexity and running time $\O{\minnk n}$ (Theorem~\ref{thm:offline}).
\end{theorem}

In the \emph{single-pass streaming model}, the algorithm is given $K$ and a stream $E$ consisting of elements $e_1,\ldots, e_n \in U$, which arrive sequentially. 
The objective is to minimize the auxiliary space used by algorithm throughout the execution. 
In the \emph{multi-pass streaming model}, the algorithm is further allowed to make multiple passes over $E$. 
This model is typically used for modeling storage devices with sequential access (e.g. hard drives) while using a small amount of RAM. 
In this setting minimizing the number of passes becomes another key priority. 
Note that since $\Omega(\minnk)$ is a trivial lower bound on space and $\Omega(n)$ is a trivial lower bound on time and query complexity of any approximation algorithm that queries feasible sets, our next result is almost optimal in most parameters of interest.

\begin{theorem}[Multi-pass streaming algorithm \smax]\label{thm:smax}
Let $\minnk = \min(n, K)$. 
There exists a multi-pass streaming algorithm \smax (Algorithm~\ref{alg:smax}) that uses $\O{\minnk}$ space and $\O{1/\eps}$ passes over the stream and outputs a $(\half-\epsilon)$-approximation to the submodular maximization problem under a knapsack constraint, with query complexity and running time\footnote{Note that when $\frac{1}{\eps}\ll K$, in terms of running time our streaming algorithm is more efficient than our offline algorithm. Hence, in the offline setting one can use the best of the two algorithms depending on the parameters.} $\O{n(1/\eps+\log \minnk)}$ \ifarxiv(see Theorem~\ref{thm:streaming})\fi.
\end{theorem} 

We also give an algorithm in the massively-parallel computation (MPC) model \cite{KarloffSV10} used to model MapReduce/Spark-like systems. 
We use the most restrictive version, which only allows linear total memory, running time and communication per round~\cite{AndoniNOY14}. 
In this model, the input set $E$ of size $n$ is arbitrarily distributed across $m$ machines, each with $s = \Oin{n/m}$ memory so that the overall memory is $\O{n}$. 
A standard setting of parameters for submodular optimization is $m = \sqrt{n/\minnk}$ and $s = \Oin{\sqrt{n\minnk}}$ (see e.g.~\cite{LiuV19,AvdiukhinMYZ19}). 
One of the machines is designated as the \emph{central machine} and outputs the solution in the end. 
The machines communicate to each other in a number of synchronous rounds. 
In each round, each machine receives an input of size $\Oin{\sqrt{n \minnk}}$, performs a local linear-time computation, and sends an output of size $\Oin{\sqrt{n \minnk}}$ to other machines before the next round begins. 
The primary objective in this model is minimizing the number of rounds. 
Our main result in this model is given below. 
\begin{theorem}[MPC algorithm \dmax]\label{thm:dmax}
Let $\minnk = \min(n, K)$. There exists an MPC algorithm 
\ifarxiv\dmax (Algorithm~\ref{alg:dmax})\fi 
that runs in $\O{1/\eps}$ rounds on $\sqrt{n/\minnk}$ machines, each with $\Oin{\sqrt{n\minnk}}$ memory. 
Each machine uses query complexity and runtime $\Oin{\sqrt{n \minnk}}$ per round. 
The algorithm outputs a $(\half-\eps)$-approximation to the submodular maximization problem under a knapsack constraint\ifarxiv (see Theorem~\ref{thm:distributed})\fi.
\end{theorem}
In particular, our algorithm uses execution time $\Oin{\sqrt{n \minnk}/\eps}$ and total communication, CPU time and number of queries $\O{n/\eps}$. 

\subsection{Relationship to previous work}\label{sec:previouswork}
The classic version of the problem considered in this work sets $c(e) = 1$ for all $e \in U$ and is known as monotone submodular maximization under a cardinality constraint and has been extensively studied. 
The celebrated result of~\cite{NemhauserWF78} gives a $1 - \nicefrac{1}{e}\approx 0.63$-approximation using \greedy, which is optimal unless $P \neq NP$,~\cite{Feige98}. 
The problem of maximizing a monotone submodular function under a knapsack constraint was introduced by~\cite{Wolsey82}, who gave an algorithm with $\approx0.35$-approximation. 
\cite{KhullerMN99} gave a simple \gormax algorithm with $1 - \nicefrac{1}{\sqrt{e}}\approx0.39$-approximation as well as a more complicated algorithm \partenum which requires a partial enumeration over an initial seed of three items and hence runs in $\O{\minnk n^4}$ time. 
\partenum was later analyzed by~\cite{Sviridenko04} who showed a $(1-\nicefrac{1}{e})\approx 0.63$-approximation, matching the hardness of~\cite{Feige98}. 
\ifarxiv
The subsequent search for more efficient algorithms has motivated a number of further studies.
\cite{BadanidiyuruV14} and~\cite{EneN17} give algorithms with approximation $1 - \nicefrac{1}{e}-\eps$. 
However while these algorithms are theoretically interesting, they are self-admittedly impractical due to their exponential dependence on large polynomials in $1/\eps$. 
\fi

Compared to the well-studied cardinality-constrained case, streaming literature on monotone submodular optimization under a knapsack constraint is relatively sparse.
A summary of results in the streaming setting is given in Figure~\ref{fig:streaming}. 
Prior to our work, the best results in streaming are by~\cite{HuangKY17,HuangK18}. 
While the most recent work of \cite{HuangK18} achieves the $(\nicefrac{1}{2}-\eps)$-approximation, its space, runtime and query complexities are far from optimal and depend on large polynomials of $1/\eps$, making it impractical for large data. 
Compared to this result, our Theorem~\ref{thm:smax} gives an improvement on all main parameters of interest, leading to near-optimal results. 
On the other hand, for the cardinality-constrained case, an optimal single-pass $(\nicefrac{1}{2}-\eps)$-approximation has very recently been achieved by~\cite{KazemiMZLK19}. 
While using different ideas, our multi-pass streaming result matches theirs in terms of approximation, space and improves slightly on the number of queries and runtime (from $\O{{n \log \minnk}/{\eps}}$ to $\O{n (\nicefrac{1}{\eps} + \log \minnk)}$) only at the cost of using a constant number of passes for constant $\eps$.

\begin{figure*}[!htb]
    \centering
    \rowcolors{1}{}{gray!10}
    \begin{tabular}{|c|c|c|c|c|}\hline
     Reference & Approx. & Passes & Space & Runtime and Queries \\\hline\hline
     \cite{HuangKY17} & $\nicefrac{1}{3}-\eps$ & 1 & $\O{\frac{1}{\eps}K\log K}$ & $\O{\frac{1}{\eps}n\log K}$ \\[1.05ex]
     \cite{HuangKY17} & $\nicefrac{4}{11}-\eps$ & 1 & $\O{\frac{1}{\eps}K\log K}$ & $\O{\frac{1}{\eps}n\log K}$ \\[1.05ex]
     \cite{HuangKY17} & $\nicefrac{2}{5}-\eps$ & 3 & $\O{\frac{1}{\eps^2}K\log^2 K}$ & $\O{\frac{1}{\eps}n\log K}$ \\[1.05ex]
     \cite{HuangK18} & $\nicefrac{1}{2}-\eps$ & $\O{\nicefrac{1}{\eps}}$ & $\O{\frac{1}{\eps^7}K\log^2 K}$ & \O{\frac{1}{\eps^8}n\log^2 K} \\[1.05ex]
     \smax (Alg.~\ref{alg:smax}) & $\nicefrac{1}{2}-\eps$ & $\O{\nicefrac{1}{\eps}}$ & $\O{K}$ & $\O{n\left(\frac{1}{\eps}+\log K\right)}$\\[1.05ex]\hline
    \end{tabular}
    \caption{Monotone submodular maximization under a knapsack constraint in the streaming model.}
    \label{fig:streaming}
\end{figure*}

In the distributed setting, \cite{MirzasoleimanKSK13} give an elegant two round protocol for monotone submodular maximization subject to a knapsack constraint that achieves a subconstant guarantee. 
\cite{KumarMVV15} later give algorithms for both cardinality and matroid constraints that achieve a constant factor approximation, but the number of rounds is $\Theta(\log\Delta)$, where $\Delta$ is the maximum increase in the objective due to a single element, which is infeasible for large datasets since $\Delta$ even be significantly larger than the size of the entire dataset.
\cite{BarbosaENW15,BarbosaENW16} subsequently give a framework for both monotone and non-monotone submodular functions under cardinality, matroid, and $p$-system constraints. 
Specifically, the results of~\cite{BarbosaENW16} achieves almost $\nicefrac{1}{2}$-approximation using two rounds, a result subsequently matched by Liu and Vondr\'{a}k without requiring the duplication of items, as well as a $(1-\nicefrac{1}{e}-\eps)$ approximation using $\O{\nicefrac{1}{\eps}}$ rounds. 
\cite{BarbosaENW15} also gives a two-round algorithm for a knapsack constraint that achieves roughly $0.17$-approximation in expectation. 

For extensions to other constraints, non-monotone objectives and other generalizations see e.g.~\cite{ChakrabartiK14,ChekuriGQ15,ChanHJKT17,ElenbergDFK17,EpastoLVZ17,MirzasoleimanJK18,FeldmanKK18,ChekuriQ19}.

\subsection{Our techniques}\label{sec:techniques}
Let $\margain{e}{S} = f(e \cup S) - f(S)$ be the \emph{marginal gain} and $\marden{e}{S} = \margain{e}{S}/c(e)$ be the \emph{marginal density} of $e$ with respect to $S$. 
\greedy starts with an empty set $G$ and repeatedly adds an item that maximizes $\marden{e}{G}$ among the remaining items that fit. 
While by itself this does not guarantee any approximation, the classic result of~\cite{KhullerMN99} shows that \gormax algorithm, which takes the best of the greedy solution and the single item with maximum value, gives a 0.39-approximation but cannot go beyond 0.44-approximation. 
Our algorithm \gmax (Algorithm~\ref{alg:gmax}) instead attempts to augment every partial greedy solution with the item giving the largest marginal gain. 
For each $i$, let $\calG_i$ be the set of the first $i$ items taken by greedy. 
We augment this solution with the item $s_i$ which maximizes $\margain{s_i}{\calG_i}$ among the remaining items that fit. 
\gmax then outputs the best solution among such augmentations.

Our main technical contribution lies in the analysis of this algorithm and its variants, which shows a $\half$-approximation (this analysis is tight\ifarxiv{, see Example~\ref{ex:tight-example}} \fi).
Let $\best$ be the item from $\opt$ with the largest cost.
The main idea is to consider the last partial greedy solution such that $\best$ still fits.
Since $\best$ has the largest cost in $\opt$, we can augment the partial solution with any element from $\opt$, and all of them have a non-greater marginal density than the next selected item.
While \gmax augments partial solutions with the best item, for the sake of analysis it suffices to consider only augmentations with $\best$ (note that the item itself is unknown to the algorithm). 

To simplify the presentation, in the analysis we rescale $f$ and the costs so that $f(\opt) = 1$ and $K = 1$.
Suppose that at some point, the partial greedy solution has collected elements with total cost $x \in [0,1]$. 
We use a continuous function $\gval(x)$ to track the performance of \greedy. 
We also introduce a function $\bigHval(x)$ to track the performance of augmentation with $\best$ and then show that $\gval$ and $\bigHval$ satisfy a differential inequality $\bigHval(x)+ (1 - \cbest)\gval'(x)\geq 1$~(Lemma~\ref{lem:hEq}),
where $\gval'$ denotes the right derivative. 
To give some intuition about the proof, consider the case when there exists a partial greedy solution of cost exactly $1 - \cbest$.
If $\bigHval(1 - \cbest) \ge \half$, then the augmenation with $\best$ gives a $\half$-approximation. Otherwise, by the differential inequality, $\gval'(1 - \cbest) \ge \nicefrac 1 {2 (1 - \cbest)}$.
Since $\gval(0)=0$ and $\gval'$ is non-increasing, $\gval(1 - \cbest) \ge (1 - \cbest) \gval'(1 - \cbest) \ge \half$. See full analysis for how to handle the cases when there is no partial solution of cost exactly $1 - \cbest$. 

Our streaming algorithm \smax and distributed algorithm \dmax approximately implement \gmax in their respective settings.
\smax makes $\O{\nicefrac{1}{\eps}}$ passes over the data, and for each pass it selects items with marginal density at least a threshold $\frac{c f (\opt)}{K (1 + \eps)^i}$ in the $i$-th pass for some constant $c > 0$. 
This requires having a constant-factor approximation of $f(\opt)$ which can be computed using a single pass. 
\dmax combines the thresholding approach with the sampling technique developed by~\cite{LiuV19} for the cardinality constraint. 
The differential inequality which we develop for \gmax turns out to be robust to various sources of error introduced through thresholding and sampling. 
As we show, it continues to hold with functions and derivatives replaced with their $(1+\eps)$-approximations, which results in $(\half-\eps)$-approximation guarantees for both algorithms.
\section{Algorithms and analysis}
\subsection{Offline algorithm \gmax}\label{sec:offline}
We introduce the main ideas by first describing our offline algorithm \gmax which is then adapted to the streaming and distributed settings. 
\ifarxiv{As this algorithm is a modification of the standard \greedy algorithm we describe \greedy first.}
\fi
\greedy starts with an empty set $G$ and in each iteration selects an item $e$ with the highest marginal density $\marden{e}{G}$ that still fits into the knapsack. 
\arxiv{We refer to the resulting solution as the \emph{greedy solution} and denote it as $G$. } 
\gmax is based on augmenting each partial solution constructed by \greedy with the item of the largest marginal value (as opposed to density) and taking the best among such augmentations. 
Recall that $\calG_i$ is the set of the first $i$ items in the greedy solution.
\gmax finds for each $i$ an augmenting item $s_i$ which maximizes $f(s_i \cup \calG_i)$ among all items that still fit\arxiv{, i.e. $c(s_i \cup G_i) \le K$}. The final output is the best among all such augmented solutions.
Implementation is given as Algorithm~\ref{alg:gmax}.
\begin{algorithm}
\textbf{Input}: Set of elements $E=e_1,\ldots, e_n$, knapsack capacity $K$, cost function $c(\cdot)$, non-negative monotone submodular function $f$\;
\textbf{Output}: $\frac{1}{2}$-approximation for submodular maximization under knapsack constraint\;
$G\gets\emptyset,S\gets\emptyset$\;
\While{$E\neq\emptyset$}{
$s\gets\argmax_{e\in E}\margain{e}{G}$\;
\If{$f(S)<f(G\cup s)$}{
$S\gets G\cup s$;
}
$a\gets\argmax_{e\in E}\marden{e}{G}$\;
$G\gets G\cup a$\;
$K\gets K-c(a)$\;
Remove all elements $e\in E$ with $c(e)>K$\;
}
\Return $S$
\caption{Offline algorithm $\gmax$}\label{alg:gmax}
\end{algorithm}
In the rest of this section we 
\ifarxiv
show
\fi
that \gmax gives $\nicefrac12$-approximation. \arxiv{This analysis is tight as illustrated by the following example:
\begin{example}\label{ex:tight-example}
Let $e_1, e_2, e_3$ be three items such that $f(e_1) = f(e_2) = \frac12$ and $f(e_3) = \frac12 + \eps$ for any $\eps > 0$.
Let $c(e_1) = c(e_2) = \frac12$ and $c(e_3) = \frac{1 + \eps}{2}$.
Let $f$ be a linear function, i.e. $f(S) = \sum_{e \in S} f(e)$. Then $\opt = \{e_1, e_2\}$ has value $1$ while \gmax outputs  $\{e_3\}$ of value $\frac12 + \eps$.
\end{example}}
\ifarxiv
As discussed in Section~\ref{sec:techniques}, our analysis is based on a number of differential inequalities for functions tracking the performance of our algorithm. We assume that these functions are continuous and piecewise smooth, and by $\xi'(x)$ we denote the right-hand derivative of $\xi$ at point $x$. All these inequalities are of the form $\xi(x) + \alpha \xi'(x) \ge \beta$ for some function $\xi$,  applied in a certain range $[u,v]$ and have some initial condition $\xi(u)$. We frequently need to integrate these inequalities to get a lower bound on $\xi(v)$ which can be done as follows:

Our proof proceeds by case analysis on whether $\best$, the item of the largest cost in $\opt$, is included in the greedy solution $G$ or not.  We first show that if $\best\in G$, then $f(G)$ is at least a $\nicefrac{1}{2}$-approximation.
\fi

\arxiv{Let $\opt$ be the optimal solution, i.e. the maximizer of $f(\opt)$ under $c(\opt)\le K$. Let $\best$ be the element of the largest cost in $\opt$. }
W.l.o.g. and only for \arxiv{the sake of} analysis of approximation we rescale the function values and costs so that $f(\opt) = 1$ and $c(\opt) = K = 1$\footnote{Note that if $c(\opt) < K$ then we can set $K = c(\opt)$ first as this does not affect $f(OPT)$.}. 
We first define a greedy performance function $\gval(x)$ which allows us to track the performance of the greedy solution in a continuous fashion. 
Let \arxiv{$G$ be the greedy solution computed by Algorithm~\ref{alg:gmax} and let} $g_1, g_2, \ldots, g_m$ be the elements in $G$ in the order they were added and recall that $\calG_i = \{g_1,\ldots,g_i\}$. 
For a fixed $x$, let \arxiv{its \emph{greedy index}} $i$ be the smallest index such that $c(\calG_i) > x$. 

\begin{definition}[Greedy performance function]
For $x \in [0,1]$ we define $\gval(x)$ as: $$\gval(x)=f(\calG_{i-1})+(x-c(\calG_{i-i}))\marden{g_i}{\calG_{i-1}}.$$
\end{definition}

\arxiv{Note that $\gval$ is a continuous and monotone piecewise-linear function such that $\gval(0)=0$. 
Since an important role in the analysis is played by the derivative of this function we further define $\gval'$ to be the right derivative for $\gval$ so that $\gval'$ is defined everywhere on the interval $[0,c(G))$ and is always non-negative. }

\ifarxiv
We now define a function $\smallhval(x)$ which tracks the performance of $\gmax$ when the greedy solution collects a set of cost $x$. 
Note that the cost of the last item which \gmax uses to augment the solution does not count in the argument of this function. 
\begin{definition}[$\gmax$ performance function]
For any fixed $x$, let $i$ be the smallest index such that $c(\calG_i)>x$. 
We define $\smallhval(x)=\gval(x)+\margain{v}{\calG_{i - 1}}$, where 
\[v=\argmax_{e\in E \setminus \calG_{i - 1}: c(e\cup\calG_{i - 1})\le K} \margain{e}{\calG_{i - 1}}\]
is the element with the largest marginal gain with respect to the current partial greedy solution $\calG_{i - 1}$.
\end{definition}

For technical reasons which we describe below instead of working directly with $\smallhval$ it is easier to work with a lower bound on it $\bigHval$ which has some nicer properties.
\fi
For 
$\bigHval$ we only consider adding $\best$, the largest item from $\opt$, to the current partial greedy solution. \arxiv{Note that hence $\bigHval$ is only defined while this item still fits.}
Consider the last item added by the greedy solution before the cost of this solution exceeds $1 - c(\best)$.
We define $\clast$ so that $1 - c(\best) - c^*$ is the cost of the greedy solution before this item is taken.

\begin{definition}[$\gmax$ performance lower bound]
For $x \in [0, 1 - \cbest - \clast] $ we define $\bigHval(x) = \gval(x)+\margain{\best}{\calG_{i - 1}}$ so that $\bigHval(x)\le\smallhval(x)$.
\end{definition}

\begin{lemma}[\gmax inequality]\label{lem:hEq}
Let $g'$ denote the right derivative of $g$.
Then for all $x \in [0, 1-c\left(\best\right)-\clast]$, the following differential inequality holds:
\begin{align*}
\bigHval(x)+(1-c\left(\best\right))\gval'(x)\geq 1
\end{align*}
\end{lemma}

\begin{proof}
Similarly to the proof of the standard greedy inequality it suffices to show the statement only for points where $x = c(\calG_{i - 1})$ for some $i \ge 1$. 
Hence, we have $\bigHval(x) = \gval(c(\calG_{i -1})) + \margain{\best}{\calG_{i - 1}} = f(\calG_{i - 1} \cup \best)$.
Since we normalized $f(\opt)=1$, then by monotonicity, $1 = f(\opt)\le f(\calG_{i - 1} \cup \opt)$. 
Hence:
\begin{align*}
1&\le f(\calG_{i - 1} \cup \opt)\\
&= f(\calG_{i - 1} \cup \best) +  \margain{\opt\setminus(\best\cup \calG_{i - 1})}{\calG_{i - 1} \cup \best}\\
&\leq  \bigHval(x)    +  \sum_{e\in\opt\setminus(\best\cup \calG_{i - 1})}\margain{e}{\calG_{i - 1} \cup \best}\\
&=  \bigHval(x)    +  \sum_{e\in\opt\setminus(\best\cup \calG_{i - 1})} c(e)\marden{e}{\calG_{i - 1} \cup \best},
\end{align*}
where the second inequality is by submodularity and the definition of $\bigHval$ and the last equality is by the definition of marginal density. 
Since $x\le 1-c\left(\best\right)-\clast$, then all items in $\opt \setminus (\best \cup \calG_{i - 1})$ still fit, as $\best$ is the largest item in $\opt$. 
Since the greedy algorithm always selects the item with the largest marginal density, then $\max_{e \in \opt \setminus (\best \cup \calG_{i - 1})} \marden{e}{\calG_{i - 1} \cup \best}\le \gval'(x)$. Hence:
\begin{align*}
1&\le \bigHval(x)    +  \sum_{e\in\opt\setminus(\best\cup \calG_{i - 1})} c(e)\marden{e}{\calG_{i - 1}\cup \best}\\
&\le \bigHval(x)    +  \sum_{e\in\opt\setminus(\best\cup \calG_{i - 1})} c(e)\marden{e}{\calG_{i - 1}}\\
&\leq  \bigHval(x)    +  \sum_{e\in\opt\setminus(\best\cup \calG_{i - 1})} c(e) \gval'(x)\\
&=  \bigHval(x)    +  \gval'(x)\sum_{e\in\opt\setminus(\best\cup \calG_{i - 1})} c(e) \\
&=  \bigHval(x)    +  \gval'(x) c(\opt\setminus(\best\cup \calG_{i - 1})) \\
&\leq  \bigHval(x)    +  \gval'(x) (1-c\left(\best\right)),
\end{align*}
where the last inequality follows from the normalization of $c(\opt)\le 1$ and the fact that $\best \in \opt$.
\end{proof}

\begin{theorem}
\label{thm:offline}
Recall that $\minnk = \min(n, K)$ is an upper bound on the number of elements in feasible solutions.
Then \gmax gives a $\nicefrac{1}{2}$-approximation to the submodular maximization problem under a knapsack constraint and runs in $\O{\minnk n}$ time.
\end{theorem}
\begin{proof}
By applying Lemma~\ref{lem:hEq} at the point $x = 1 - \cbest - \clast$, we have:
\[\bigHval(1 - \cbest - \clast) + (1-c\left(\best\right))\gval'(1 - \cbest - \clast)\geq 1\]
If $\bigHval(1 - c(\best) - \clast) \ge \frac 1 2$, then we have $\frac 12$-approximation, because $\bigHval(1 - c(\best) - \clast)$ is a lower bound on the value of the augmented solution when the cost of the greedy part is $1 - c(\best) - \clast$. Otherwise: 
\begin{align*}
\gval'(1 - \cbest - \clast) &\geq \frac {1 - \bigHval(1 - \cbest - \clast)} {1 - \cbest}> \frac 1 {2 (1 - \cbest)}.
\end{align*}

Note that since $\gval(0) = 0$ and $\gval'$ is non-increasing by the definition of \greedy, for any $x \in [0, 1]$ we have $\gval(x) \ge \gval'(x)\cdot x$:
\[\gval(x) \ge \int_{\chi=0}^x \gval'(\chi) d\chi \ge \int_{\chi=0}^x \gval'(x) d\chi = \gval'(x) \cdot x,\]
Therefore, applying this inequality at $x = 1 - \cbest - \clast$:
\begin{align*}
\gval(1 - \cbest - \clast)
&\ge (1 - \cbest - \clast) \gval'(1 - \cbest - \clast)\\
&\ge \frac {1 - \cbest - \clast} {2 (1 - \cbest)}.
\end{align*}

Recall that $1 - \cbest - \clast$ was the last cost of the greedy solution when we could still augment it with $\best$; therefore, the next element $e$ that the greedy solution selects has the cost at least $(1 - \cbest) - (1 - \cbest - \clast) = \clast$. 
Thus, the function value after taking $e$ is at least
\begin{align*}
\gval(1 - \cbest - \clast) + \clast \gval'(1 - \cbest - \clast)\ge \frac {1 - \cbest - \clast} {2 (1 - \cbest)} + \frac {c^*} {2 (1 - \cbest)} = \frac 1 2
\end{align*}

Hence, Algorithm~\ref{alg:gmax} gives a $\frac{1}{2}$-approximation to the submodular maximization problem under a knapsack constraint.
It remains to analyze the running time and query complexity of Algorithm~\ref{alg:gmax}. 
Since $\minnk$ is the maximum size of a feasible set, Algorithm \ref{alg:gmax} makes at most $\minnk$ iterations. 
In each iteration, it makes $\O{n}$ oracle queries, so the total number of queries and runtime is $\O{\minnk n}$.
\end{proof}
\subsection{Streaming algorithm \smax}\label{sec:streaming}
Our multi-pass streaming algorithm is given as Algorithm~\ref{alg:smax}. 
To simplify the presentation, we first give the algorithm under the assumption that it is given a parameter $\lambda$, which is a constant-factor approximation of $f(\opt)$. 
We then show how to remove this assumption using standard techniques in
\ifarxiv
Theorem~\ref{thm:constant:stream}.
\fi 
As discussed in the description of our techniques \smax uses $\O{1/\eps}$ passes over the data to simulate the execution of \gmax approximately.

\begin{algorithm}
	\textbf{Input}: Stream $e_1,\ldots, e_n$, knapsack capacity $K$, cost function $c(\cdot)$, non-negative monotone submodular function $f$, $\lambda$ which is an $\alpha$-approximation of $f(\opt)$ for some fixed constant $\alpha$>0, $\eps>0$\;
	\textbf{Output}: $(\nicefrac{1}{2}-\epsilon)$-approx. for submodular maximization under a knapsack constraint\;
	$T\gets\emptyset$, $\tau\gets\frac{\lambda}{\alpha K}$\;
	\While(\tcp*[f]{Thresholding stage}){$\tau>\frac{\lambda}{2K}$}{    
		Take a new pass over the stream\; 
		\For{each read item $e$}{
			\If{$\marden{e}{T}\geq\tau$ and $c(e\cup T)\leq K$}{
				$T\gets T\cup \{e\}$\;
			}
		}
		$\tau\gets\tau/(1+\epsilon)$\;
	}
	For each $i$, let $G_i$ be the first $i$ selected in the construction of $T$ above and let $s_i=\emptyset$\;
	Take a pass over the stream\;
	\For(\tcp*[f]{Augmentation stage}){each read item $e$}{
		\If{$e\notin T$}{
			$j = \max\{i|c(G_i) + c(e) \le K\}$\;
			\If{$f(G_j\cup s_j)<f(G_j\cup e)$}{
				$s_j\gets \{e\}$\;
			}
		}
	}
	\Return $\argmax f(G_i\cup s_i)$
	\caption{Multi-pass streaming algorithm $\smax$}\label{alg:smax}
\end{algorithm}

\ifarxiv
In the analysis, which gives the proof of Theorem~\ref{thm:smax}, we define functions $\tval$\arxiv{, $\smallrval$,} and $\bigRval$ analogous to $\gval$\arxiv{, $\smallhval$,} and $\bigHval$ respectively, based on $\calT_i$, the first $i$ items collected by the thresholding algorithm.
We show that $\tval$ and $\bigRval$ satisfy the same differential inequalities as $\gval$ and $\bigHval$ respectively, up to $(1+\eps)$ factors\arxiv{, and similar to before, our analysis then proceeds by casework on whether $\best$, the largest item in $\opt$, is included in the thresholding solution $T$ or not}. 

We first show that if $\best\in T$, then $f(T)$ is at least a $\left(\frac{1}{2}-\eps\right)$-approximation.
\fi

Let $T$ be the set of items constructed $\smax$ (as in Algorithm~\ref{alg:smax}) and let $t_1,t_2,\ldots$ be the order that they are collected. We refer to the part of the algorithm which constructs $T$ as ``thresholding'' and the rest as ``augmentation'' below. 
We use $\calT_i$ to denote the set containing the $i$ items $\{t_1,t_2,\ldots,t_i\}$.
We again use $\best$ to denote the item with highest cost in $\opt$.
Similar to the above, we define two functions representing the values of our thresholding algorithm, and augmented solutions given the utilized proportion of the knapsack.

\begin{definition}[Thresholding performance function]
For any $x \in [0,1]$, let $i$ be the smallest index such that $c(\calT_i)>x$. 
We define $\tval(x)=f(\calT_{i-1})+(x-c(\calT_{i-i}))\marden{t_i}{\calT_{i-1}}$\arxiv{ and $t'(x)$ to be the right derivative of $t$}.
\end{definition}

We define a function $\bigRval(x)$ that lower bounds the performance of $\smax$ when the thresholding solution collects a set of cost $x$:

\begin{definition}[$\smax$ performance function and lower bound]
For any fixed $x$, let $i$ be the smallest index such that $c(\calT_i)>x$. 
Then we define $\bigRval(x) = \tval(x)+\margain{\best}{\calT_{i-1}}$, where $\best = \argmax_{e \in \opt} c(e)$.
\end{definition}

In order to analyze the output of the algorithm, we prove a differential inequality for $\bigRval$ rather than $\smallrval$. If $c(T) \ge 1 - c(\best)$ then let $\clast \ge 0$ be defined so that $1 - c(\best) - \clast$ is the cost of the thresholding solution before the algorithm takes the item which makes the cost exceed $1 - c(\best)$.

\begin{lemma}[\smax Inequality]
\label{lem:APXhEq}
If $c(T) \ge 1 - c(\best)$ then for all $x \in [0, 1-\cbest-\clast]$, then $\tval$ and $\bigRval$ satisfy the following differential inequality: 
\[\bigRval(x)+(1+\eps)(1-c\left(\best\right))\tval'(x)\geq 1.\]
\end{lemma}

\ifarxiv
\begin{proof}
First, note that for $x \in [0,p]$ where $p$ is the total cost of items taken in the first pass the inequality holds trivially since $\tval'(x) \ge 1$ (as in the proof of the standard thresholding inequality).
Hence assume that $x \in [p, 1 - c(\best)-\clast]$ is fixed and consider any pass after the first one.
Similarly to other proofs it suffices to only consider left endpoints of the intervals of the form $[c(\calT_{i - 1}), c(\calT_i))$ so let $x = c(\calT_{i - 1})$.
Since we normalized $f(\opt)=1$, then by monotonicity, $1 = f(\opt)\le f(\calT_{i-1}\cup \opt)$. 
Hence:
\begin{align*}
1&\le f(\calT_{i-1}\cup \opt)\\
& = f((\calT_{i-1} \cup \best) \cup (\opt \setminus \best))\\
&= f(\calT_{i-1}\cup \best) +  \margain{\opt\setminus(\best\cup \calT_{i-1})}{\calT_{i-1}\cup \best}\\
&\leq  \bigRval(x)    +  \sum_{e\in\opt\setminus(\best\cup \calT_{i-1})}\margain{e}{\calT_{i-1}\cup \best}\\
&=  \bigRval(x)    +  \sum_{e\in\opt\setminus(\best\cup \calT_{i-1})} c(e)\marden{e}{\calT_{i-1}\cup \best},
\end{align*}
where the second inequality is by submodularity and the last line is by the definition of marginal density. 
Since $\best$ has the maximum cost in $\opt$. $x \le 1 - \cbest$, all items in $\opt \setminus (\best \cup \calT_{i - 1})$ still fit into the remaining knapsack capacity.
In all passes after the first one, the thresholding algorithm always selects an element which gives $\frac 1 {1 + \eps}$-approximation of the highest possible marginal density:
$$(1 + \eps)\tval(x) \ge \max_{e \in \opt \setminus (\best \cup \calT_{i - 1})} \marden{e}{\best \cup \calT_{i - 1}}.$$

Combining with the inequality above:
\begin{align*}
1&\le \bigRval(x)    +  \sum_{e\in\opt\setminus(\best\cup \calT_{i-1})} c(e)\marden{e}{\calT_{i-1}\cup \best}\\
&\leq  \bigRval(x)  +  (1+\eps)\tval'(x)\sum_{e\in\opt\setminus(\best\cup \calT_{i-1})} c(e) \\
&= \bigRval(x)  +  (1+\eps)\tval'(x) c(\opt\setminus(\best\cup \calT_{i-1})) \\
&\leq \bigRval(x)  +  (1+\eps)\tval'(x) (1-\cbest),
\end{align*}
where the last equality is by the normalization of $c(\opt)=1$ and the fact that $\best \in \opt$. 
\end{proof}
\fi

\begin{theorem}
\label{thm:streaming}
There exists an algorithm that uses $\O{\minnk}$ space and $\O{\nicefrac{1}{\eps}}$ passes over the stream, makes $\O{\nicefrac{n}{\eps}+n\log \minnk}$ queries, and outputs a $\left(\nicefrac{1}{2}-\epsilon\right)$-approximation to the submodular maximization problem under a knapsack constraint. 
\end{theorem}

\ifnotarxiv
The proofs are similar to the proofs of Lemma~\ref{lem:hEq} and Theorem~\ref{thm:offline} and are provided in the supplementary material.
\fi

\ifarxiv
\begin{proof}
We can use existing algorithm from Theorem~\ref{thm:constant:stream} to obtain a constant factor approximation $\lambda$ to $f(\opt)$. 
We thus analyze the correctness of Algorithm \ref{alg:smax} given an input $\lambda$ that is a constant factor approximation to $f(\opt)$. 
The proof is similar to proof of Theorem~\ref{thm:offline}. 

By applying Lemma~\ref{lem:APXhEq} at the point $x = 1 - \cbest - \clast$, we have:
\[\bigRval(1 - \cbest - \clast) + (1+\eps)(1-\cbest)\tval'(1 - \cbest - \clast)\geq 1\]
If $\bigRval(1 - \cbest - \clast) \ge \frac 1 2$, then we have $\frac 12$-approximation, because $\bigRval(1 - c(\best) - \clast)$ is a lower bound on the value of the augmented solution when the cost of the thresholding solution is $1 - c(\best) - \clast$. Otherwise: 
\begin{align*}
\tval'(1 - \cbest - \clast) &\geq \frac {1 - \bigHval(1 - \cbest - \clast)} {(1 - \cbest)(1 + \eps)} \\
&> \frac 1 {2 (1 - \cbest)(1 + \eps)}
\end{align*}

Note that since $\tval(0) = 0$, for any $x \in [0, 1]$ we have $\tval(x) \ge \frac {\tval'(x)\cdot x} {1 + \eps}$:
\[\tval(x) \ge \int_{\chi=0}^x \tval'(\chi) d\chi \ge \int_{\chi=0}^x \frac {\tval'(x)} {1 + \eps} d\chi = \frac {\tval'(x) \cdot x} {1 + \eps},\]
where we used the fact that $\tval'$ is a $\frac 1 {1 + \eps}$ approximation of the maximum marginal density, which does not increase. 
Therefore, applying this at $x = 1 - \cbest - \clast$:
\begin{align*}
\tval(1 - \cbest - \clast) &\ge (1 - \cbest - \clast) \tval'(1 - \cbest - \clast)\\
&\ge \frac {1 - \cbest - \clast} {2 (1 - \cbest) (1 + \eps)}.
\end{align*}

Recall that $1 - \cbest - \clast$ was the last cost of the thresholding solution when we could still augment it with $\best$; therefore, the next element $e$ that the thresholding solution selects has the cost at least $(1 - \cbest) - (1 - \cbest - \clast) = \clast$.
Thus, the function value after taking $e$ is at least
\begin{align*}
\gval(1 - \cbest - \clast) + \clast \gval'(1 - \cbest - \clast) &\ge \frac {1 - \cbest - \clast} {2 (1 - \cbest) (1 + \eps)} + \frac {c^*} {2 (1 - \cbest) (1 + \eps)}\\
&= \frac 1 {2 (1 + \eps)} = \frac 12 -\frac \eps {2 (1 + \eps)} \ge \frac{1}{2}-\eps.
\end{align*}

Hence, Algorithm~\ref{alg:smax} gives a $\left(\frac{1}{2}-\eps\right)$-approximation to the submodular maximization problem under knapsack constraints, given a constant factor approximation to $f(\opt)$. 
Note that it suffices to consider only thresholds up to $\frac{\tau}{2K}$ since $\tval'(x) < \frac{1}{2}$ implies that $\tval(x)>\frac{1}{2}$ by Lemma~\ref{lem:APXgEq}. 

Using existing algorithms to obtain a constant factor approximation $\lambda$ (e.g., by setting $\eps=\frac{1}{6}$ in Theorem~\ref{thm:constant:stream}) that use additional $\O{n\log \minnk}$ queries, then correctness of Algorithm~\ref{alg:smax} follows. 
It remains to analyze the space and query complexity of Algorithm \ref{alg:smax}.
Since each item has cost at least $1$, at most $K$ items are stored by the thresholding algorithm, and at most $K$ items are stored by the augmented solution $S$. 
Hence, the space complexity of Algorithm~\ref{alg:smax} is $\O{\minnk}$. 
If $\tau$ is an $\alpha$-approximation to $f(\opt)$ for some constant $\alpha$, then the algorithm makes $\log_{1+\epsilon}\frac{1}{2\alpha}=\O{\frac{1}{\eps}}$ passes over the input stream. 
Each pass makes at most $n$ queries, so the number of queries is at most $\O{\frac{n}{\eps}}$. 
\end{proof}
\fi
\ifarxiv
\subsection{Distributed algorithm \dmax}\label{sec:distributed}
In this section, we assume that there are $\machines = \sqrt{\nicefrac{n}{\minnk}}$ machines $M_1,\ldots,M_{\machines}$, each with $O\left(\frac{n}{\machines}\right)=O(\sqrt{n \minnk})$ amount of local memory.
Our distributed algorithm \arxiv{(Algorithm~\ref{alg:dmax})} follows a similar thresholding approach as our streaming algorithm: at each round, machines collect items whose marginal densities exceed the threshold corresponding to the round. 
\begin{algorithm}
\textbf{Input}: Set of elements $E=e_1,\ldots, e_n$, knapsack capacity $K$, cost function $c(\cdot)$, non-negative monotone submodular function $f$, $\tau$ that is $\alpha$-approximation of $f(\opt)$ for some constant $\alpha>0$\;
\textbf{Output}: A set $S$ that is a $(\frac{1}{2}-\epsilon)$-approximation for submodular maximization with a knapsack constraint\;
$T\gets\emptyset$, $t\gets\frac{\tau}{\alpha K}, \minnk \gets \min(n, K)$\;
\While{$t>\frac{\tau}{2K}$}{
    Form $\Gamma$ by sampling each $e\in E$ with probability $4\sqrt{\minnk/n}$\;
    Partition $E$ randomly into sets $V_1, V_2, \ldots V_\machines$\;
    Send $V_i$ to machine $M_i$ for all $i$\;
    Send $\Gamma$ and $T$ to all machines including a central machine $C$\;
    \For{each machine $M_i$ (in parallel)}{
        $X_i = T$\;
        \For{each item $e\in \Gamma$}{
            \If{$\rho(e|X_i)>t$}{
                $X_i=X_i\cup\{e\}$\;
            }
        }
        \For{each item $e\in V_i$}{
            \If{$\rho(e|X_i)>t$}{
                $X_i=X_i\cup\{e\}$\;
            }
        }
        $X_i=X_i\setminus T$\;
        Send $X_i\setminus T$ to $C$\;
    }
    \For{each item $e\in\cup X_i$ (on central machine)}{
        \If{$\rho(e|T)>t$}{
            $T=T\cup\{e\}$\;
        }
    }
    $t=\frac{t}{1+\eps}$
}
Send $T$ to all machines\;
\For{each machine $M_i$ (in parallel)}{
    For each $i$, let $G_i$ denote the first $i$ items that a greedy algorithm would select from $T$ and initialize $s_i=\emptyset$\;
    \For{each item $e\in V_i \setminus T$}{
        $j = \max\set{i|c(e) + c(G_i) \le K}$\;
        \If{$f(G_j\cup s_j)<f(G_j\cup e)$}{
            $s_j\gets e$\;
        }
    }
    Send $\argmax f(G_i\cup s_i)\}$ to $C$\;
}
\Return $\argmax$ of solutions received in $C$
\caption{$\dmax$: A $O\left(\frac{1}{\epsilon}\right)$-round MapReduce algorithm for sub-modular maximization under knapsack constraints.}
\label{alg:dmax}
\end{algorithm} 
\fi

\ifarxiv
We require the following form of Azuma's inequality for submartingales.
\begin{theorem}[Azuma's Inequality]
\label{thm:azuma}
Suppose $X_0, X_1,\ldots,X_n$ is a submartingale and $|X_i-X_{i+1}|\le c_i$. 
Then
\[\Pr\left[X_n-X_0\le -t\right]\le\eexp\left(\frac{-t^2}{2\sum_i c_i^2}\right).\]
\end{theorem}
We first bound the total number of elements sent to the central machine.
\begin{lemma}
\label{lem:total:central}
In Algorithm~\ref{alg:dmax}, with probability $1-\exp{-\Omega(K)}$, the total number of elements sent to the central machine is $\sqrt{n \minnk}$. 
\end{lemma}

\begin{proof}
Since each element is sampled with probability $\frac{\minnk}{n}$, the expected number of elements in $\Gamma_i$ is $4 \sqrt{n\minnk}$ for any round $i$.  
Hence $|\Gamma_i|\ge 3\sqrt{n\minnk}$ with probability at least $1-e^{-\Omega(\minnk)}$ by standard Chernoff bounds. 
Let $N_i$ denote the total number of elements with marginal density at least $\frac{f(\opt)}{(1+\eps)^i\minnk}$ with respect to $T_{i-1}$, so that the number of elements sent to the central unit in round $i$ is exactly $N_i + |\Gamma_i|$. 

Suppose $\Gamma_i$ is partitioned into at least $3\minnk$ chunks of size $\sqrt{\frac{n}{\minnk}}$ elements. 
If there are less than $\sqrt{n\minnk}$ remaining elements before each chunk whose marginal density with respect to $T_{i-1}$ exceeds $\frac{f(\opt)}{(1+\eps)^i\minnk}$, then certainly at most $\sqrt{n\minnk}$ elements are sent to the central machine. 

On the other hand, if there are at least $\sqrt{n \minnk}$ remaining elements before each chunk whose marginal density with respect to $T_{i-1}$ exceeds $\frac{f(\opt)}{(1+\eps)^i\minnk}$. 
Then an additional element is added to $N_i$ with probability at least $1 - \left(1 - \sqrt{\frac{\minnk}{n}}\right)^{\sqrt{\frac{n}{\minnk}}} > 1/2$. 
To use a martingale argument to bound the number of elements selected in $\Gamma_i$, we let $X_i$ be the indicator random variable for the event that at least one element is selected from the $i$\th block so that we have $E[X_i \mid X_1,\ldots,X_{i-1}] \geq\frac{1}{2}$. 
Let $Y_i = \sum_{j=1}^{i} (X_i - 1/2)$ so that the sequence $Y_1,Y_2,\ldots$ is a submartingale, i.e., $E[Y_i \mid Y_1,\ldots,Y_{i-1}] \geq Y_{i-1}$ and $|Y_{i} - Y_{i-1}| \leq 1$. 
By Azuma's inequality (Theorem~\ref{thm:azuma}), $\Pr[Y_{3\minnk} < -\frac12 \minnk] < e^{-\Omega(\minnk)}$, so that $\sum_{j=1}^{3\minnk} X_j = Y_K + \frac{3}{2}\minnk \geq \minnk$ with probability at least $1-e^{-\Omega(\minnk)}$, in which case no elements are sent to the central machine.  
\end{proof}

We now analyze the approximation guarantee and performance of Algorithm~\ref{alg:dmax}.
\fi
\ifarxiv
\begin{theorem}
\label{thm:distributed}
There exists an algorithm \dmax which uses $\O{\nicefrac{1}{\eps}}$ rounds of communication between $\sqrt{\nicefrac{n}{\minnk}}$ machines, each with $\O{\sqrt{n \minnk}}$ memory. 
With high probability, the total number of elements sent to the central machine is $\sqrt{n \minnk}$ and the algorithm outputs a $\left(\nicefrac{1}{2}-\epsilon\right)$-approximation to the submodular maximization problem with a knapsack constraint.
\end{theorem}

\begin{proof}
Correctness follows from the observation that the algorithm performs thresholding in the same manner as Algorithm~\ref{alg:smax}. 
The space bounds follow from Lemma~\ref{lem:total:central}.
\end{proof}
\fi
\subsection{Query lower bound}\label{sec:query-lb}
We show a simple query lower bound under the standard assumption \cite{NorouziFardTMZMS18,KazemiMZLK19} that the algorithm only queries $f$ on feasible sets.
\begin{theorem}\label{thm:query-lb}
For $\alpha > \nicefrac{1}{2}$, any $\alpha$-approximation algorithm for maximizing a function $f$ under a knapsack constraint that succeeds with constant probability and only queries values of the function $f$ on feasible sets (i.e. sets of cost at most $K$) must make at least $\Omega(n^2)$ queries if $f$ is either: 1) non-monotone submodular, 2) monotone and submodular on the feasible sets, 3) monotone subadditive.
\end{theorem}

\arxiv{
\begin{proof}
Let $e_1, \dots, e_n$ be the set of elements and set $c(e_i) = K/2$ for all $i$. 
By Yao's principle it suffices to consider two hard distributions $\mathcal D_{1/2}$ and $\mathcal D_{1}$ such that the optimum for every instance in the support of these distributions is $1/2$ and $1$ respectively and then show that no algorithm making $o(n^2)$ deterministic queries can distinguish the two distributions with constant probability. The distributions $\mathcal D_{1/2}$ and $\mathcal D_1$ are as follows:
\begin{itemize}
\item $\mathcal D_{1/2}$ has $f(S) = 1/2$ for all $S \neq \emptyset$. 
\item $\mathcal D_1$ is constructed by picking two items $e_i \neq e_j$ uniformly at random and assigning $f(S) = 1$ for $S=\{e_i,e_j\}$. 
Otherwise, set $f(S) = 1/2$ for all $S \neq \emptyset$ and $S\neq\{e_i,e_j\}$.
\end{itemize}

Fix the set of deterministic queries $Q$ that the algorithm makes. Since the algorithm is only allowed to make queries to sets of cost at most $K$, all sets in $Q$ have size at most two. Furthermore, note that $f(e_i) = 1/2$ for all $i$ under both $\mathcal D_{1/2}$ and $\mathcal D_1$. Thus, only queries to sets of size exactly two can help the algorithm distinguish the two distributions. All such queries give value $1/2$ under both distributions except for a single query $(i,j)$ under $\mathcal D_1$ which gives value $1$. Since $(i,j)$ is chosen uniformly at random under $\mathcal D_1$ the probability that a fixed set $Q$ contains it is given as $|Q|/\binom{n}{2}$. Hence if the algorithm succeeds with a constant probability then it must be the case that $|Q| = \Omega(n^2)$.

Note that the construction of $f$ results in a non-monotone submodular function but $f$ is monotone when restricted to feasible sets of size at most two items. 
By changing $\mathcal D_1$ so that the functions in this distribution take value $1$ on all sets of size more than $2$ one can ensure monotonicity of $f$. However, $f$ is still submodular on the feasible sets and subadditive everywhere (recall that a subadditive function satisfies $f(S) + f(T) \ge f(S \cup T)$ for all $S, T \subseteq U$).
\end{proof}
}
\section{Experimental results}
We compare our offline algorithm $\gmax$ and our streaming algorithm $\smax$ with baselines, answering the following questions: (1) What are the approximation factors we are getting on real data? (2) How do the objective values compare? (3) How do the runtimes compare? (4) How do the numbers of queries compare?
\ifarxiv
We compare $\gmax$ to the following baselines:
\begin{compactenum}
\item \textbf{\textsc{PartialEnum+Greedy}}~\cite{Sviridenko04}. Given an input parameter $d$, this algorithm creates a separate knapsack for each combination of $d$ items, and then runs the \textsc{Greedy} algorithm on each of the knapsacks. 
At the end, the algorithm outputs the best solution among all knapsacks, so that the total runtime is $\Omega(Kn^{d+1})$. 
In fact, \partenum is only feasible for $d=1$ and our smallest dataset. 
\item \textbf{\textsc{Greedy}}. This algorithm starts with an empty knapsack and repeatedly adds the item with the highest marginal density with respect to the collected items in the knapsack, until no more item can be added to the knapsack. 
\item \textbf{\textsc{GreedyOrMax}}~\cite{KhullerMN99}. This algorithm compares the value of the best item with the value of the output of the \textsc{Greedy} algorithm and outputs the better of the two.
\end{compactenum}
\fi
In streaming we compare \smax to \sieve~\cite{BadanidiyuryMKK14} and \sormax~\cite{HuangKY17}, which are similar thresholding-based algorithms.  
\ifarxiv
\sieve starts with an empty knapsack and collects all items whose marginal density with respect to the items in the knapsack exceed a given threshold (which is initially equal to $\frac{1}{2}$), while \sormax uses a similar approach, but compares the items collected by the thresholding algorithm to the best single item, and outputs the better of the two solutions. 
\fi
We also implemented a single-pass \branching by~\cite{HuangKY17} that uses thresholding along with multiple branches and gives a $\nicefrac{4}{11}\approx0.36$-approximation. 
We did not implement~\cite{HuangK18} as their algorithms are orders of magnitude slower than \branching which is already several orders of magnitude slower than other algorithms.

Our code is available at \url{https://github.com/aistats20submodular/aistats20submodular}. 
\ifarxiv
\subsection{Objectives and Datasets}
\fi

\textbf{Graph coverage.}
For a graph $G(V, E)$ and $Z\subset V$, the objective is to maximize the neighborhood vertex coverage function $f(Z) :=|Z\cup N(Z)|/|V|$, where $N(Z)$ is the set of neighbors of $Z$. 
\ifarxiv
The cost of each node is roughly proportional to the value of the node. Specifically, the cost of each node $v\in V$ is $c(v) = \frac \beta {|V|} (|N(v)| - \alpha)$, where $\alpha = \frac 1 {20}$ and $\beta$ is a normalizing factor so that $c(v) \ge 1$, so that the cost of each node is roughly proportional to the value of the node.
\fi
We ran experiments on two graphs from SNAP~\cite{snapnets}:
\ifarxiv
1) \texttt{ego-Facebook}(4K vertices, 81K edges), 2) \texttt{com-DBLP} (317K vertices, 1M edges).
\fi

\textbf{Movie ratings.}
We also analyze a dataset of movies to model the scenario of movie recommendation. 
The objective function, defined as in~\cite{AvdiukhinMYZ19}, is maximized for a set of movies that is similar to a user's interests\arxiv{ and the cost of a movie is set to be roughly proportional to its value}.
\ifarxiv
Each movie is assigned a rating in the range $[1,5]$ by users. 
Let $r_{x,u}$ be the rating assigned by user $u$ to movie $x$ and $r_{avg}$ be the average rating across all movies. 
For each movie $x$, we normalize the ratings to produce a vector $v_x$ by setting $v_{x,u}=0$ if user $u$ did not rate movie $x$ and $v_{x, u} = r_{x, u} - r_{avg}$ otherwise. 
We then define the similarity between two movies $x_1$ and $x_2$ as the dot product $\langle v_{x_1}, v_{x_2} \rangle$ of their vectors. 
Given a set $X$ of movies, to quantify how representative a subset of movies $Z$ is, we consider a parameterized objective function $f_X(Z) = \sum_{x \in X} \max_{z \in Z} \langle v_z, v_x \rangle$. 
Hence, the maximizer of $f_X(Z)$ corresponds to a set of movies that is similar to the user's interests. 
\fi
We analyze the \texttt{ml-20} MovieLens dataset~\cite{movielens}\arxiv{, which contains approximately $27K$ movies and $20M$ ratings}.

\ifarxiv
\subsection{Results}
\fi
We first give instance-specific approximation factors for different values of $K$  for offline (Fig.~\ref{fig:offline_opt_approx}) and streaming (Fig.~\ref{fig:streaming_opt_approx}) algorithms.
These approximations are computed using upper bounds on $f(\opt)$ which can be obtained using the analysis of \greedy. 
\gmax and \smax typically perform at least $20\%$ better than their $\nicefrac{1}{2}$ worst-case guarantees. 
\ifarxiv
In fact, our results show that the output value can be improved by up to $50\%$, both by \gmax upon \greedy (Figure~\ref{fig:offline_obj}) and by \smax upon \sieve (Figure~\ref{fig:streaming_obj}).
\fi
\begin{figure*}[!ht]
	\centering
	\subfloat[\texttt{com-dblp}]{\includegraphics[width=0.32\textwidth]{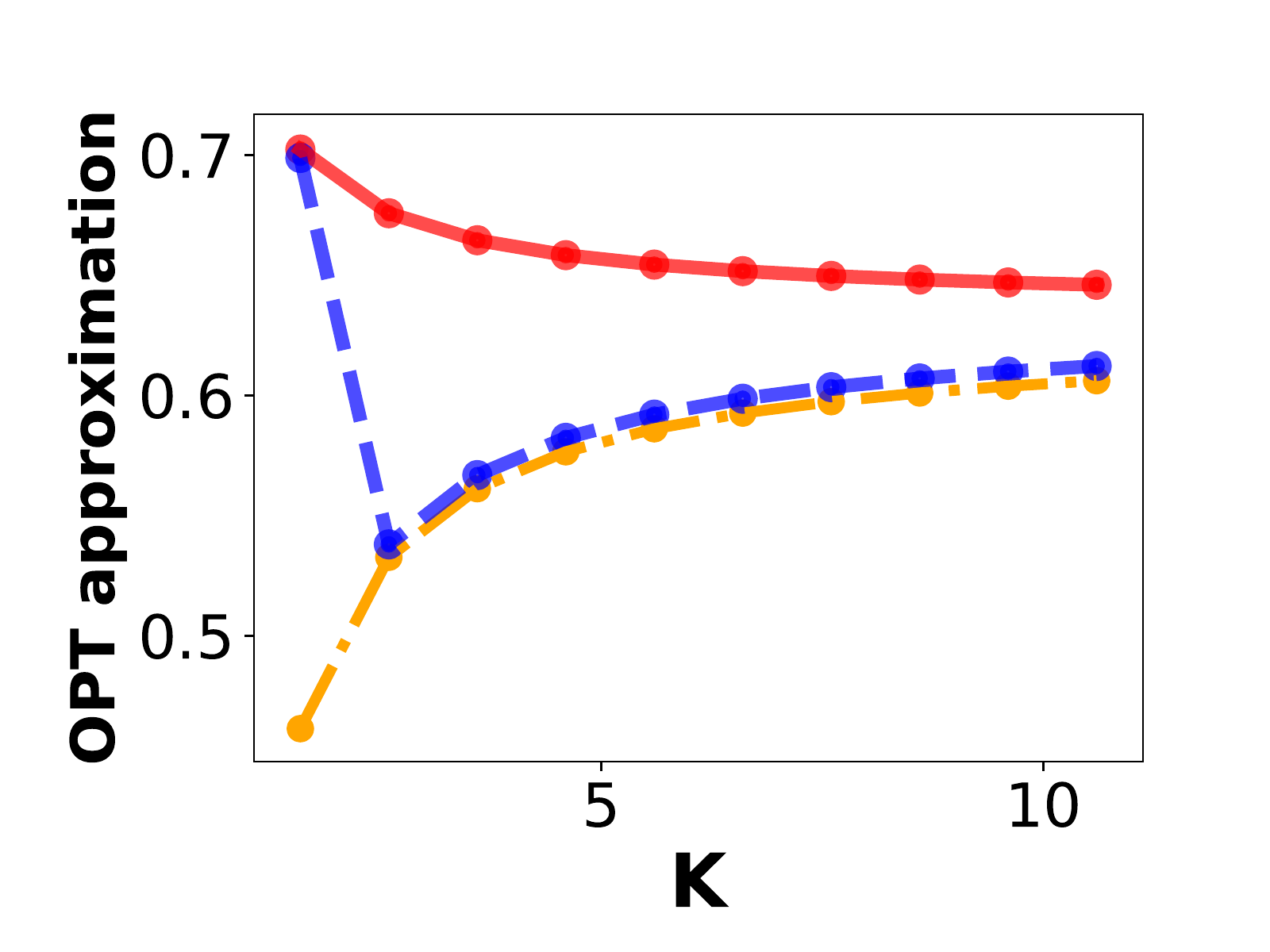}}
	\subfloat[\texttt{ego-Facebook}]{\includegraphics[width=0.32\textwidth]{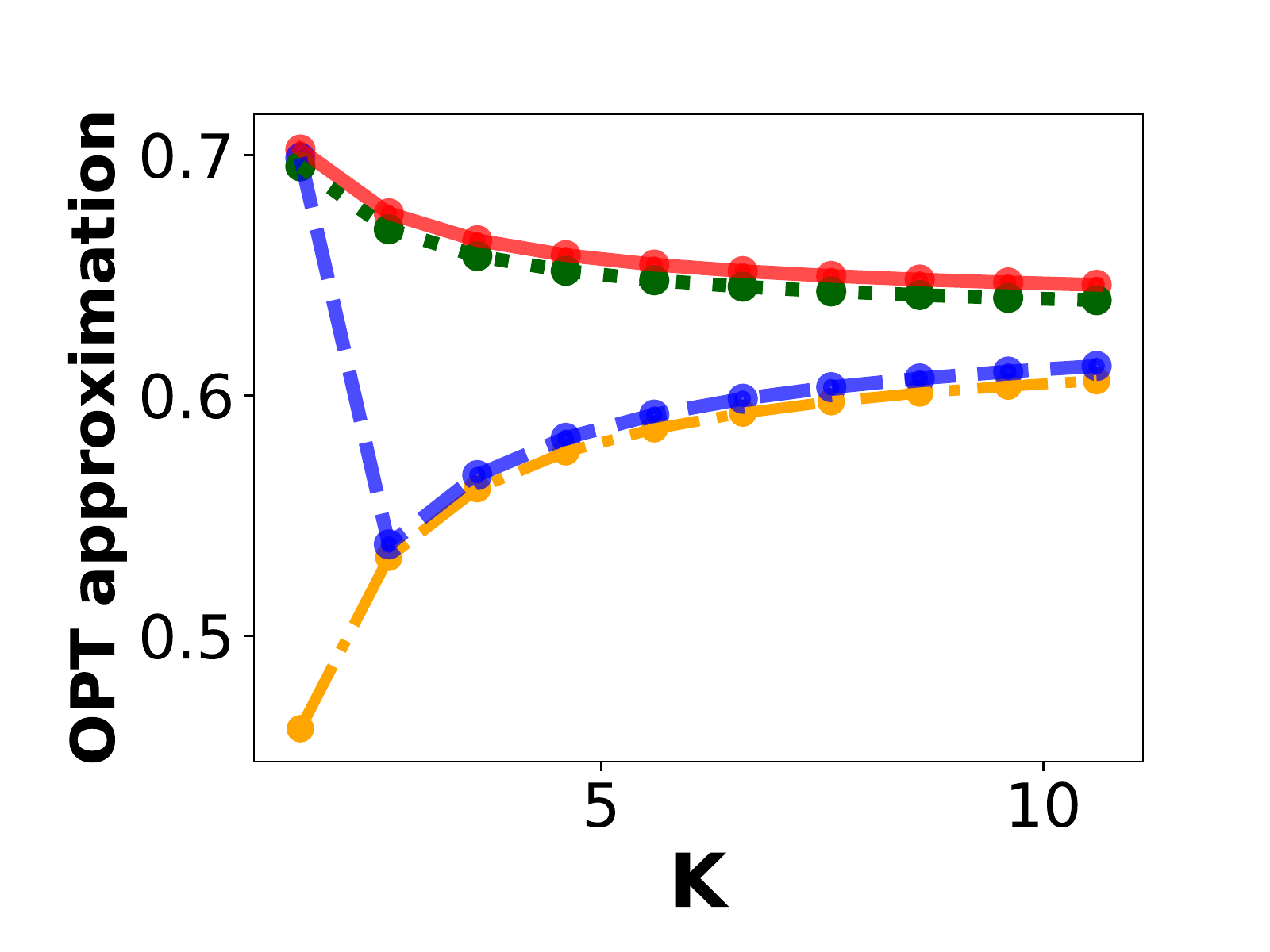}}
	\subfloat[\texttt{ml-20}]{\includegraphics[width=0.32\textwidth]{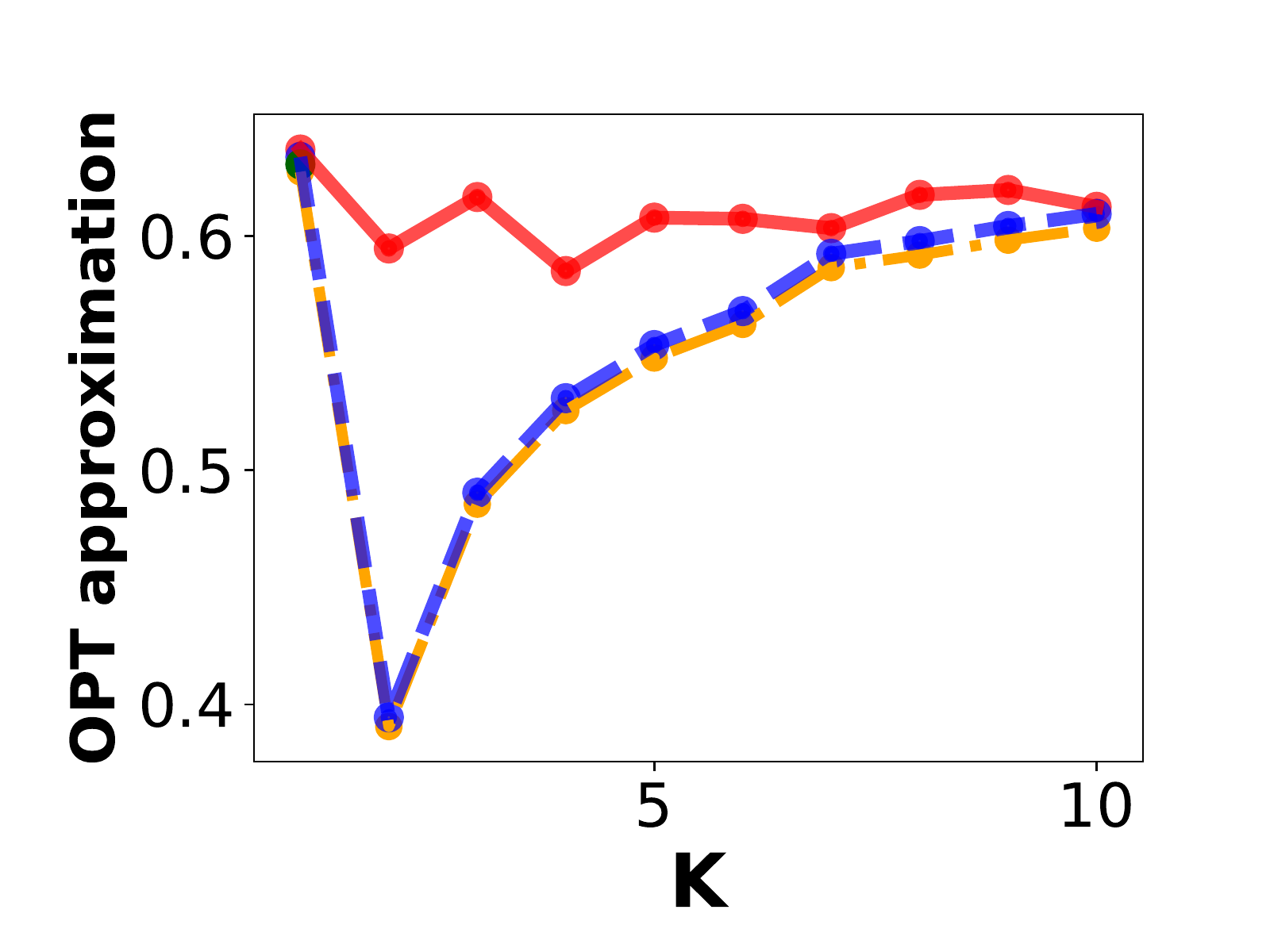}}\\
	\vspace{-0.3cm}
	\subfloat{\includegraphics[width=\textwidth]{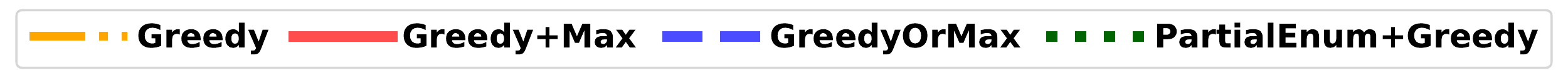}}
	\caption{Instance-specific approximations for different $K$. \gmax performs substantially better than its worst-case $\nicefrac{1}{2}$-approximation guarantee and typically beats even the $(1 - \nicefrac1e)\approx0.63$ bound. Despite much higher runtime, \partenum does not beat \gmax even on the only dataset where its runtime is feasible (\texttt{ego-Facebook}).}
	\label{fig:offline_opt_approx}
\end{figure*}
\begin{figure*}[!ht]
	\centering
	\subfloat[\texttt{com-dblp}]{\includegraphics[width=0.32\textwidth]{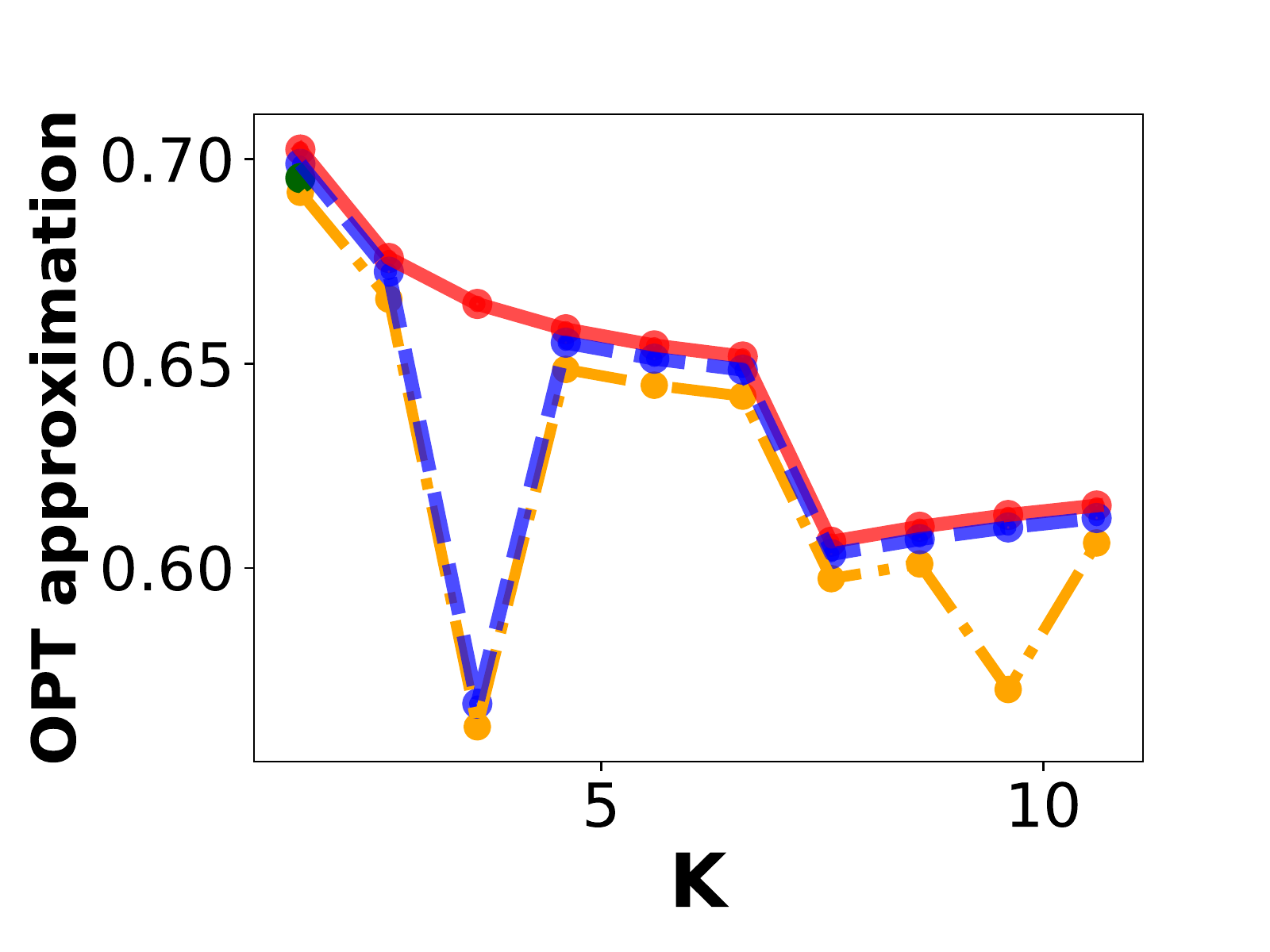}}
	\subfloat[\texttt{ego-Facebook}]{\includegraphics[width=0.32\textwidth]{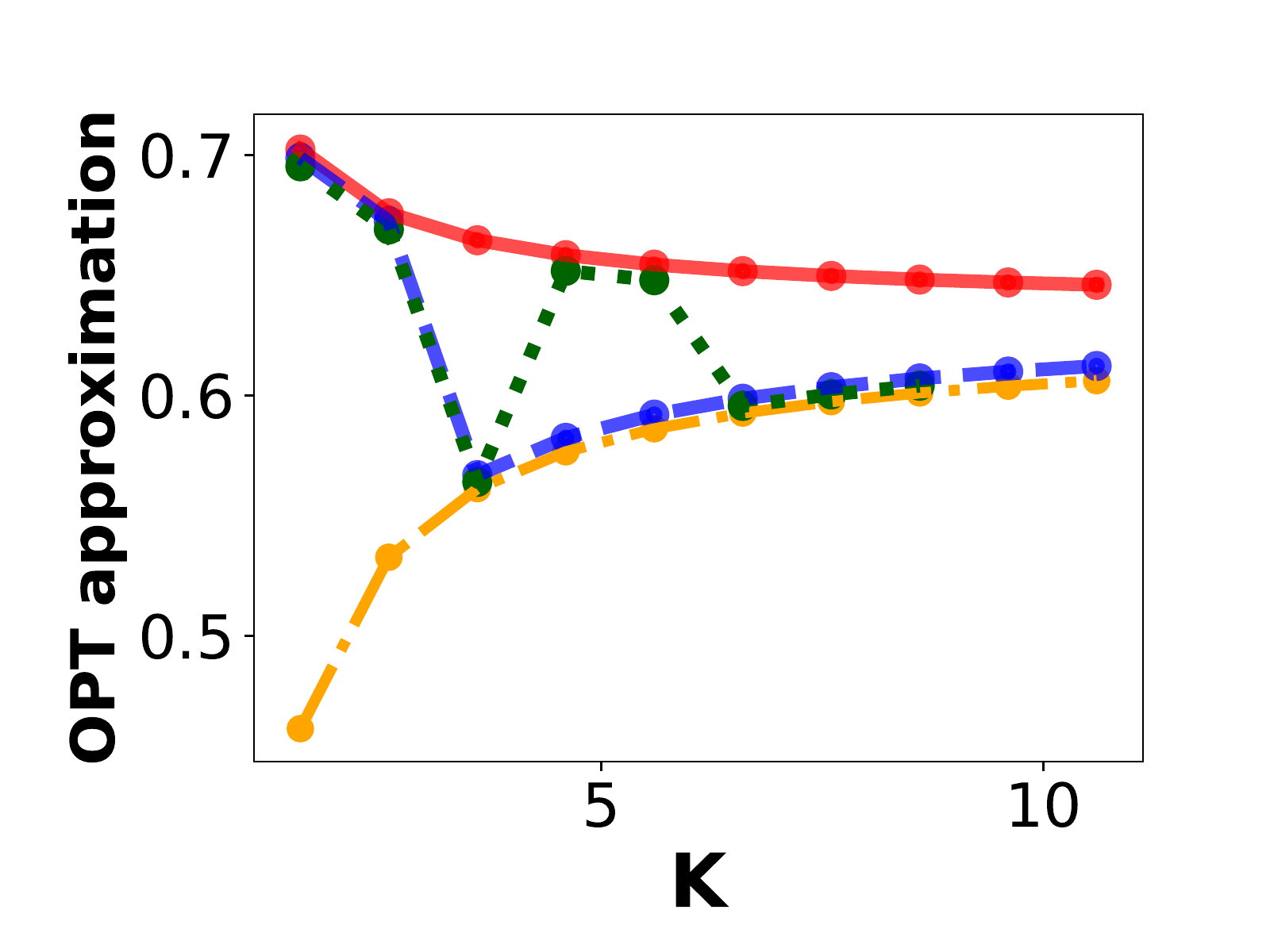}}
	\subfloat[\texttt{ml-20}]{\includegraphics[width=0.32\textwidth]{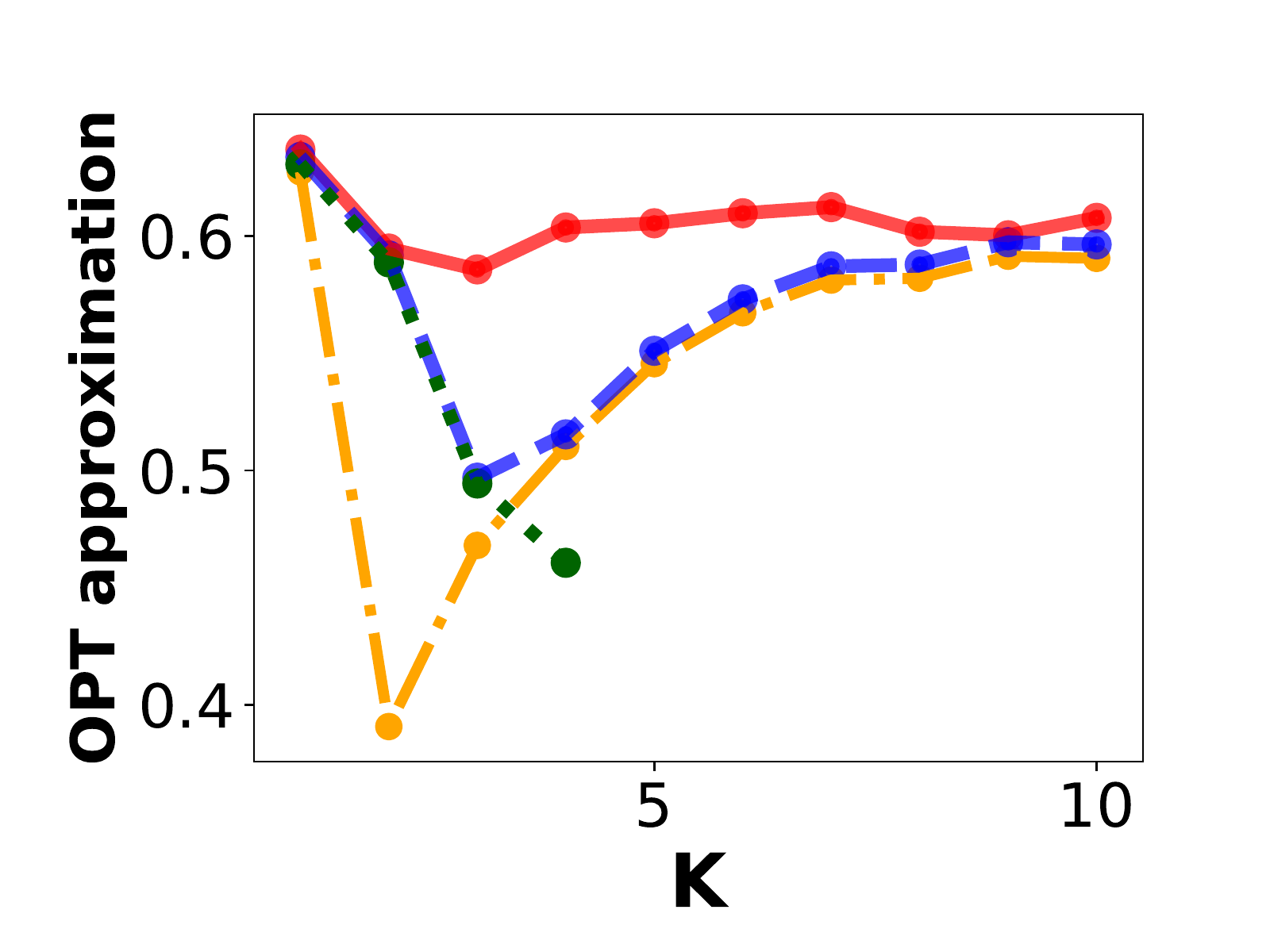}}\\
	\vspace{-0.3cm}
	\subfloat{\includegraphics[width=\textwidth]{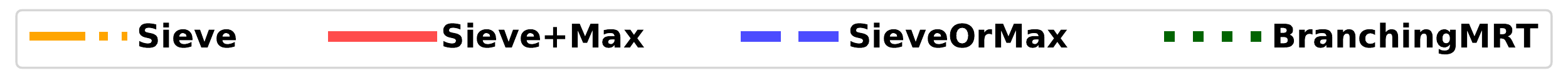}}
	\caption{Instance-specific approximations for different $K$. \smax performs substantially better than its worst-case $(\nicefrac{1}{2}-\eps)$-approximation guarantee and robustly dominates all other approaches. It can improve by up to $40\%$ upon \sieve. Despite much higher runtime, \branching does not beat \smax \arxiv{(some data points not shown for \branching as it did not terminate under a $200$-second time limit).}
	}
	\label{fig:streaming_opt_approx}
\end{figure*}

\ifarxiv	
\paragraph{Running time.} We point out that the runtimes of \gmax and \gormax algorithms are similar, being at most $20\%$ greater than the runtime of \greedy, as shown in Figure~\ref{fig:offline_times}. 
On the other hand, even though \partenum does not outperform \gmax, it is only feasible for $d=1$ and the \texttt{ego-Facebook} dataset and uses on average almost $500$ times as much runtime for $K=10$ across ten iterations of each algorithm, as shown in Figure~\ref{fig:offline_times}. 
The runtimes of \smax, \sormax, and \sieve are generally similar; however in the case of the \texttt{com-dbpl} dataset, the runtime of \smax grows with $K$.
This can be explained by the fact that oracle calls on larger sets typically require more time, and augmented sets typically contain more elements than sets encountered during execution of \sieve.
On the other hand, the runtime of \branching was substantially slower, and we did not include its runtime for scaling purposes, as for $K=5$, the runtime of \branching was already a factor 80K more than \sieve. 
Error bars for the standard deviations of the runtimes of the streaming algorithms are given in Figure~\ref{fig:streaming_times_bars}. 
\begin{figure*}[!ht]
	\centering
	\subfloat[\texttt{com-dblp}]{\includegraphics[width=0.32\textwidth]{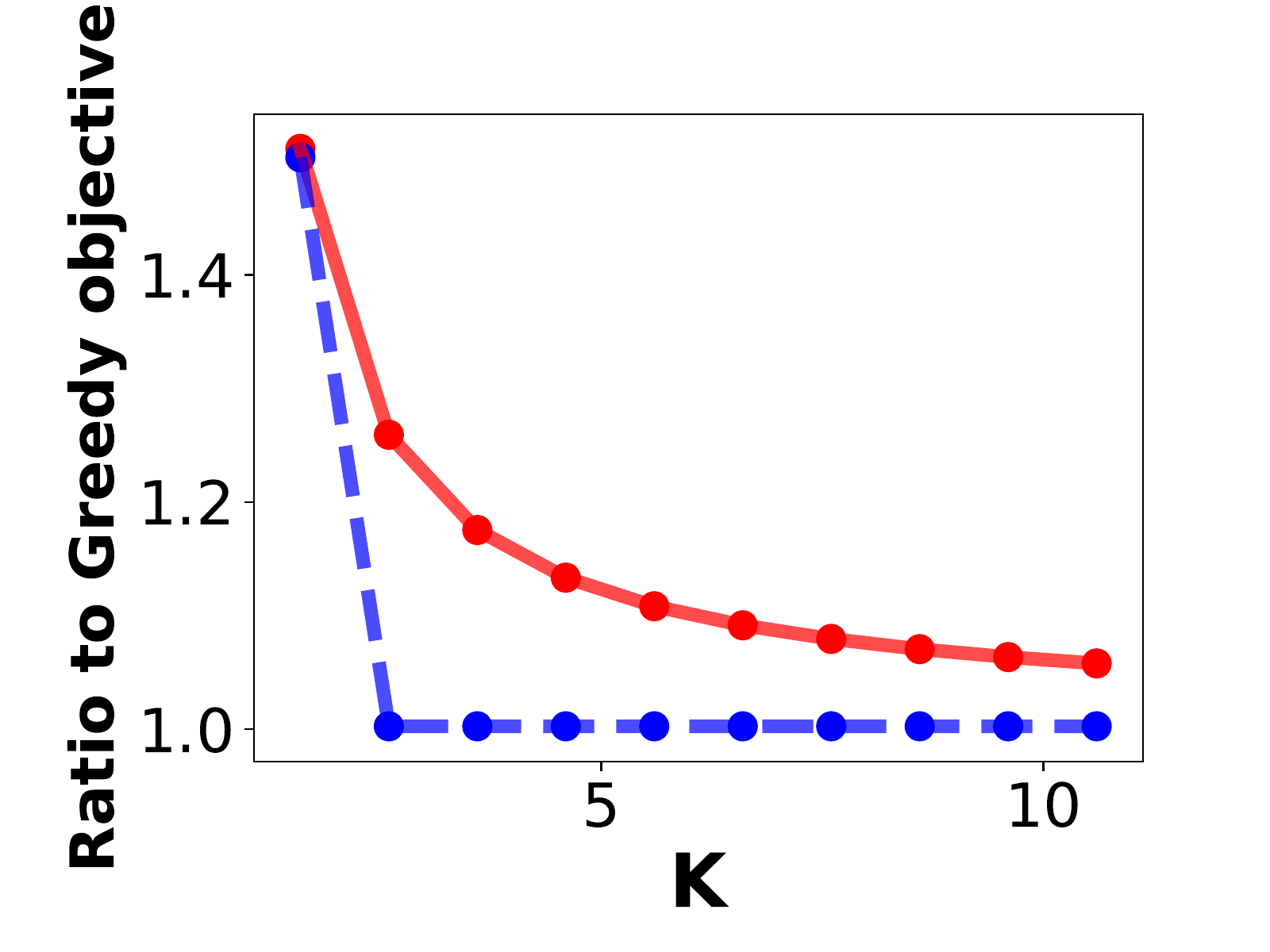}}
	\subfloat[\texttt{ego-Facebook}]{\includegraphics[width=0.32\textwidth]{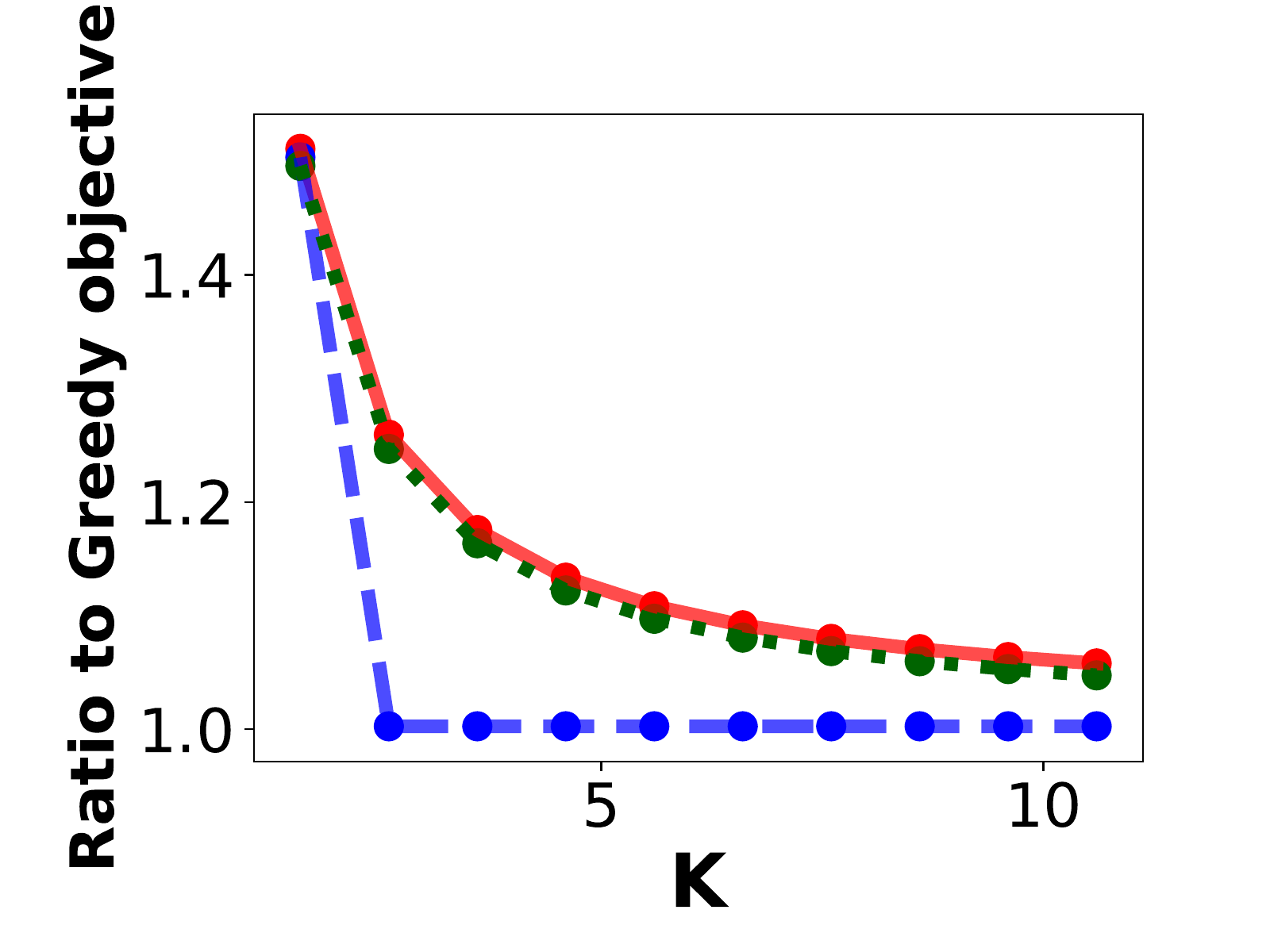}}
	\subfloat[\texttt{ml-20}]{\includegraphics[width=0.32\textwidth]{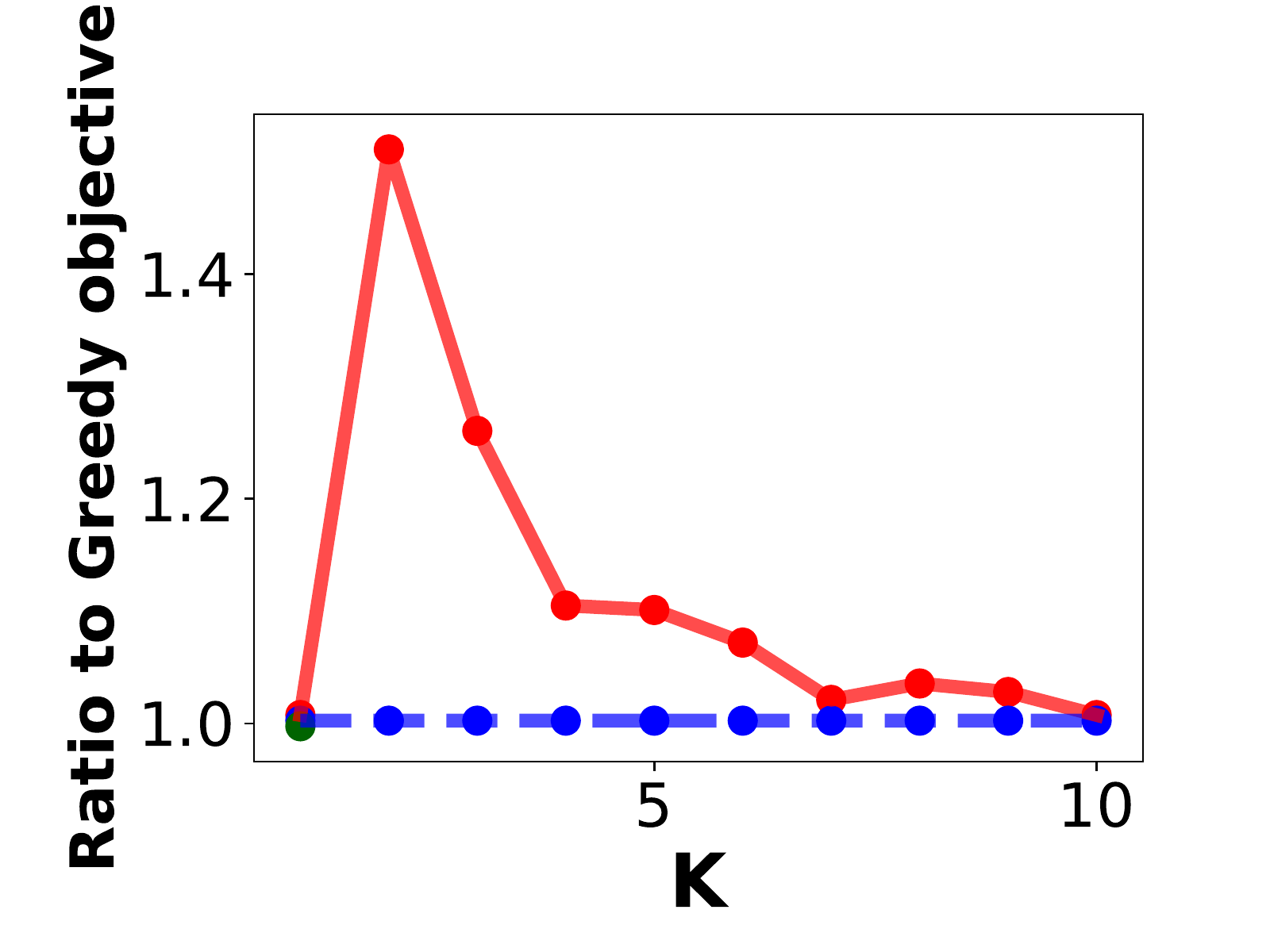}}\\
	\vspace{-0.3cm}
	\subfloat{\includegraphics[width=\textwidth]{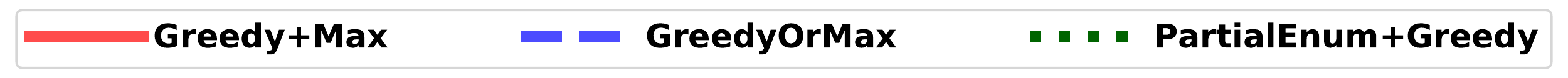}}
	\caption{Ratio of the objective of offline algorithms to the objective of \greedy for different values of $K$. \gmax can improve by almost $50\%$ upon \greedy, but by definition, \gmax and \gormax cannot perform worse than \greedy. Despite its runtime, \partenum does not outperform \gmax on the \texttt{ego-Facebook} dataset.}
	\label{fig:offline_obj}
\end{figure*}
\begin{figure*}[!ht]
	\centering
	\subfloat[\texttt{com-dblp}]{\includegraphics[width=0.32\textwidth]{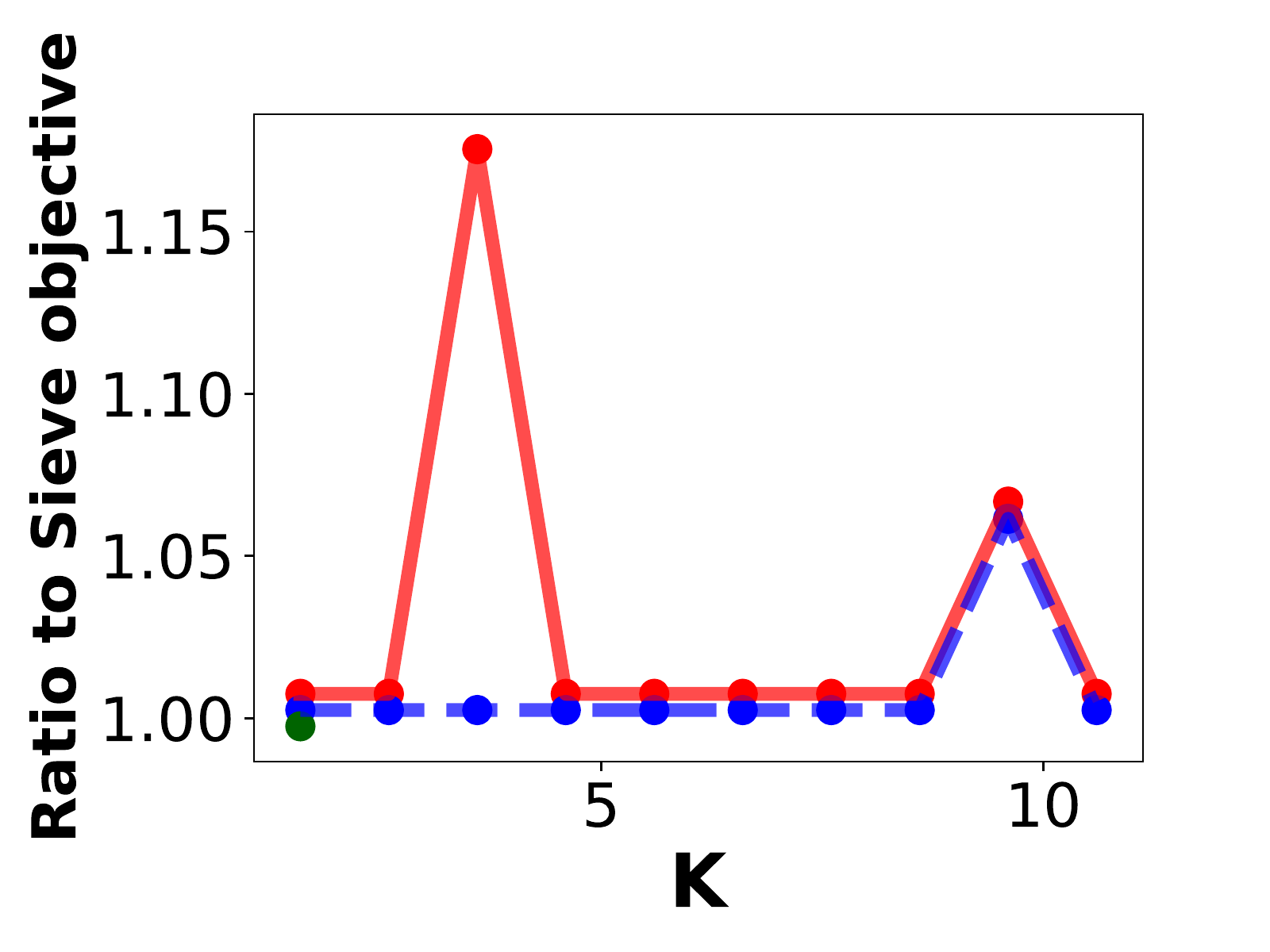}}
	\subfloat[\texttt{ego-Facebook}]{\includegraphics[width=0.32\textwidth]{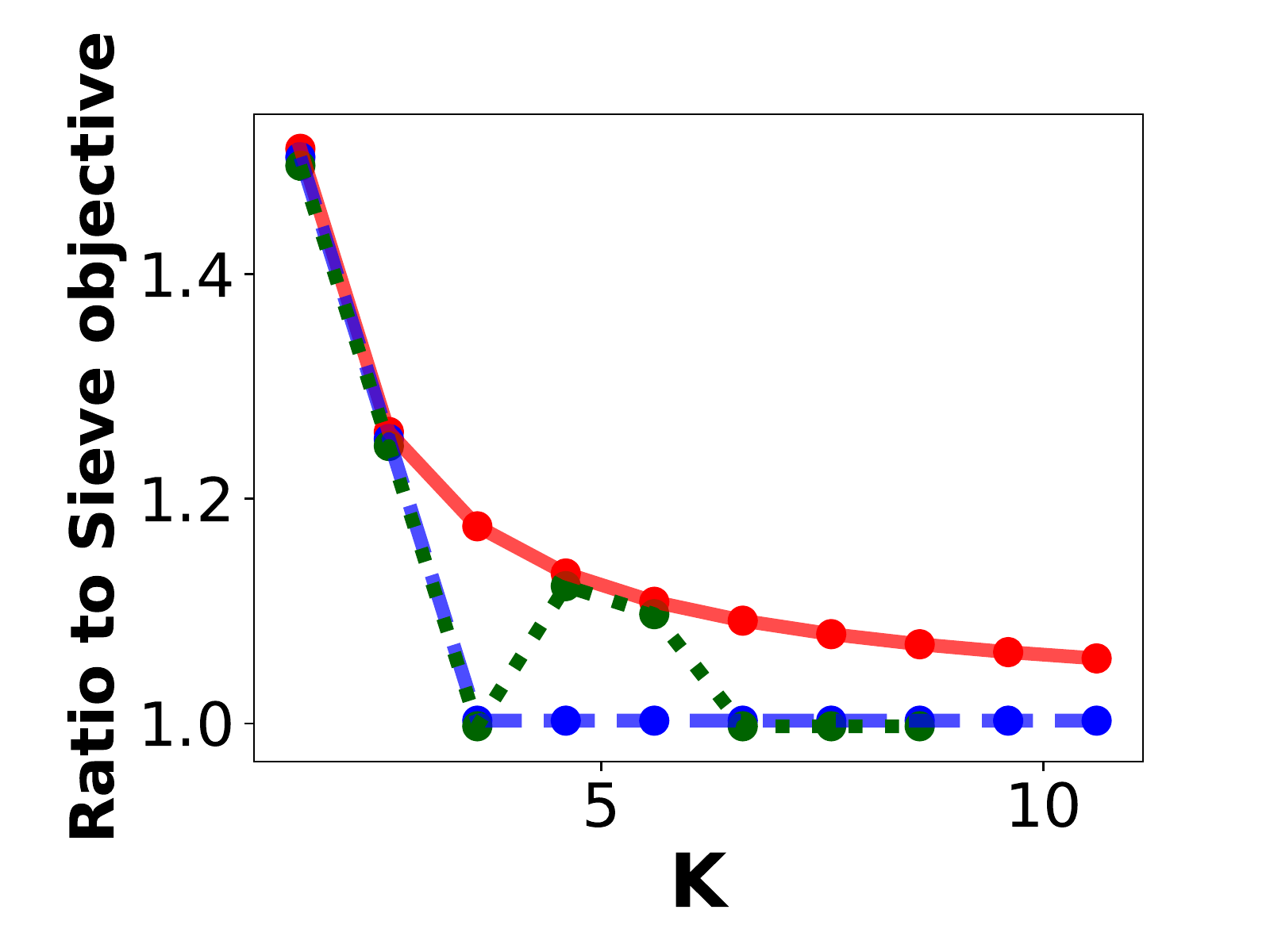}}
	\subfloat[\texttt{ml-20}]{\includegraphics[width=0.32\textwidth]{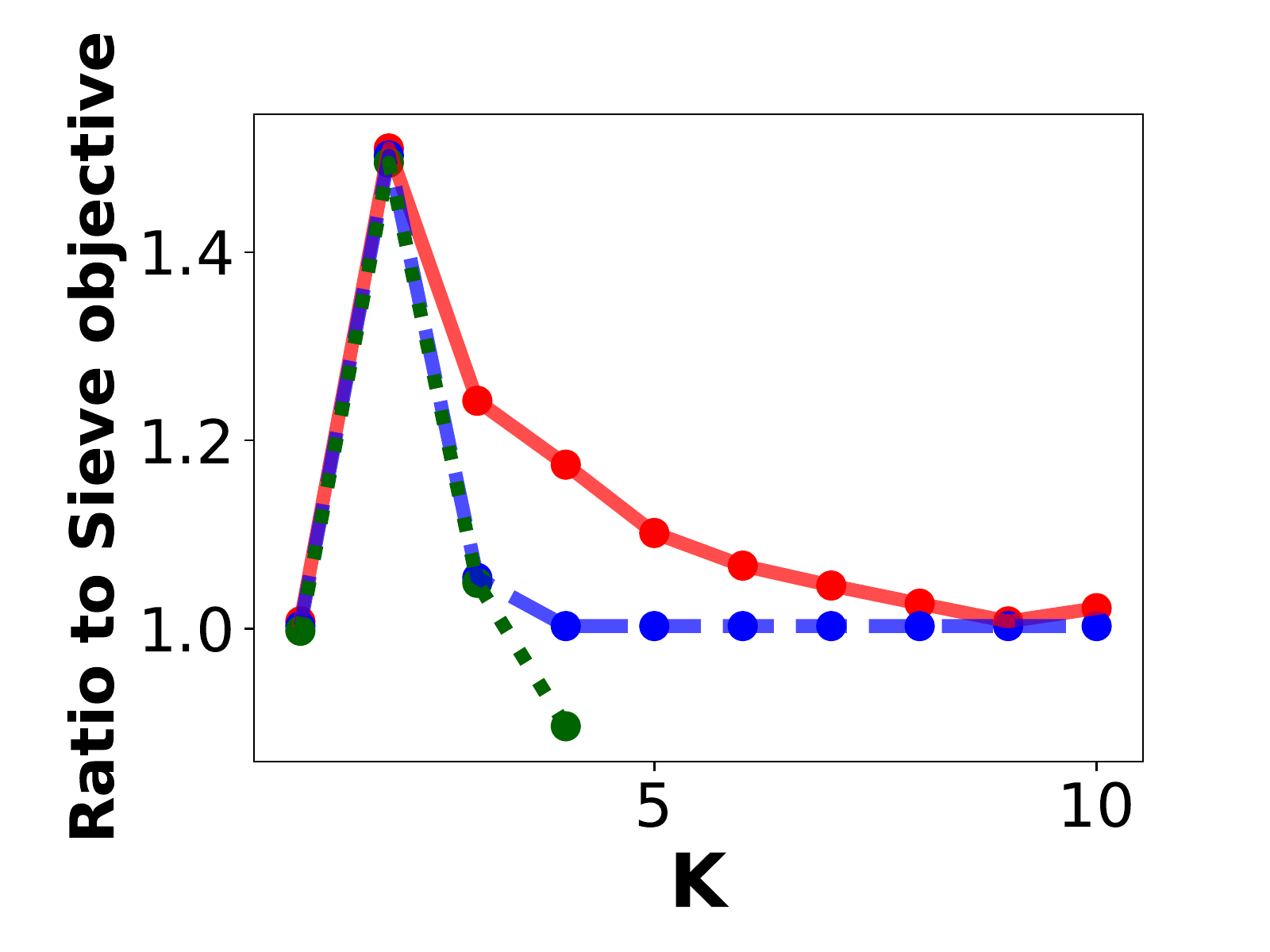}}\\
	\vspace{-0.3cm}
	\subfloat{\includegraphics[width=\textwidth]{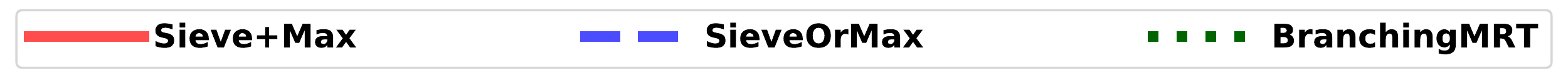}}
	\caption{Ratio of the objective of streaming algorithms to the objective of \sieve for different values of $K$. \smax can improve by almost $40\%$ upon \sieve, but by definition, \smax and \sormax cannot perform worse than \sieve. Despite its runtime, \branching does not outperform \smax.}
	\label{fig:streaming_obj}
\end{figure*}
\begin{figure*}[!ht]
	\centering
	\subfloat[\texttt{com-dblp}]{\includegraphics[width=0.32\textwidth]{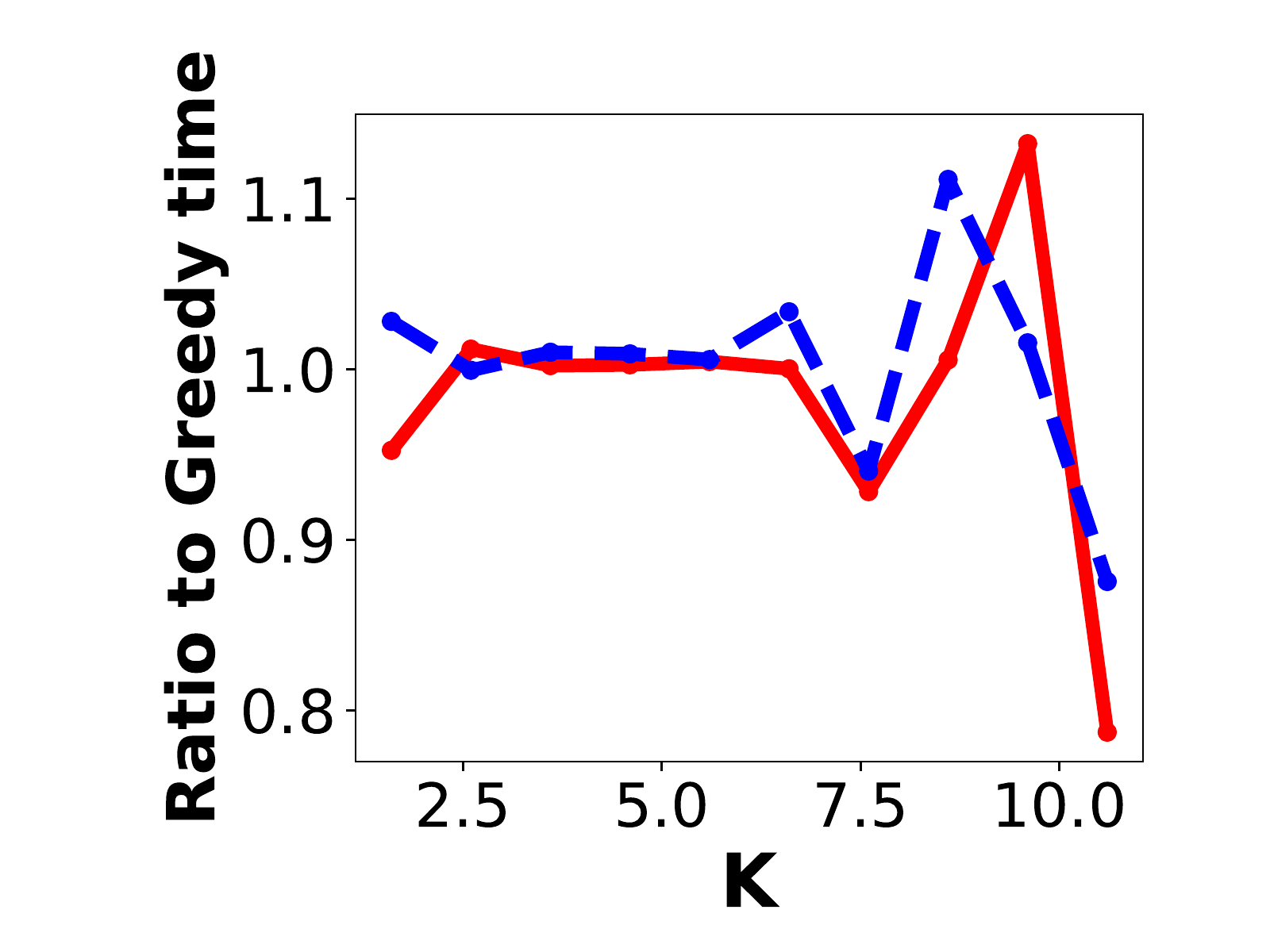}\label{fig:times_start}}
	\subfloat[\texttt{ego-Facebook}]{\includegraphics[width=0.32\textwidth]{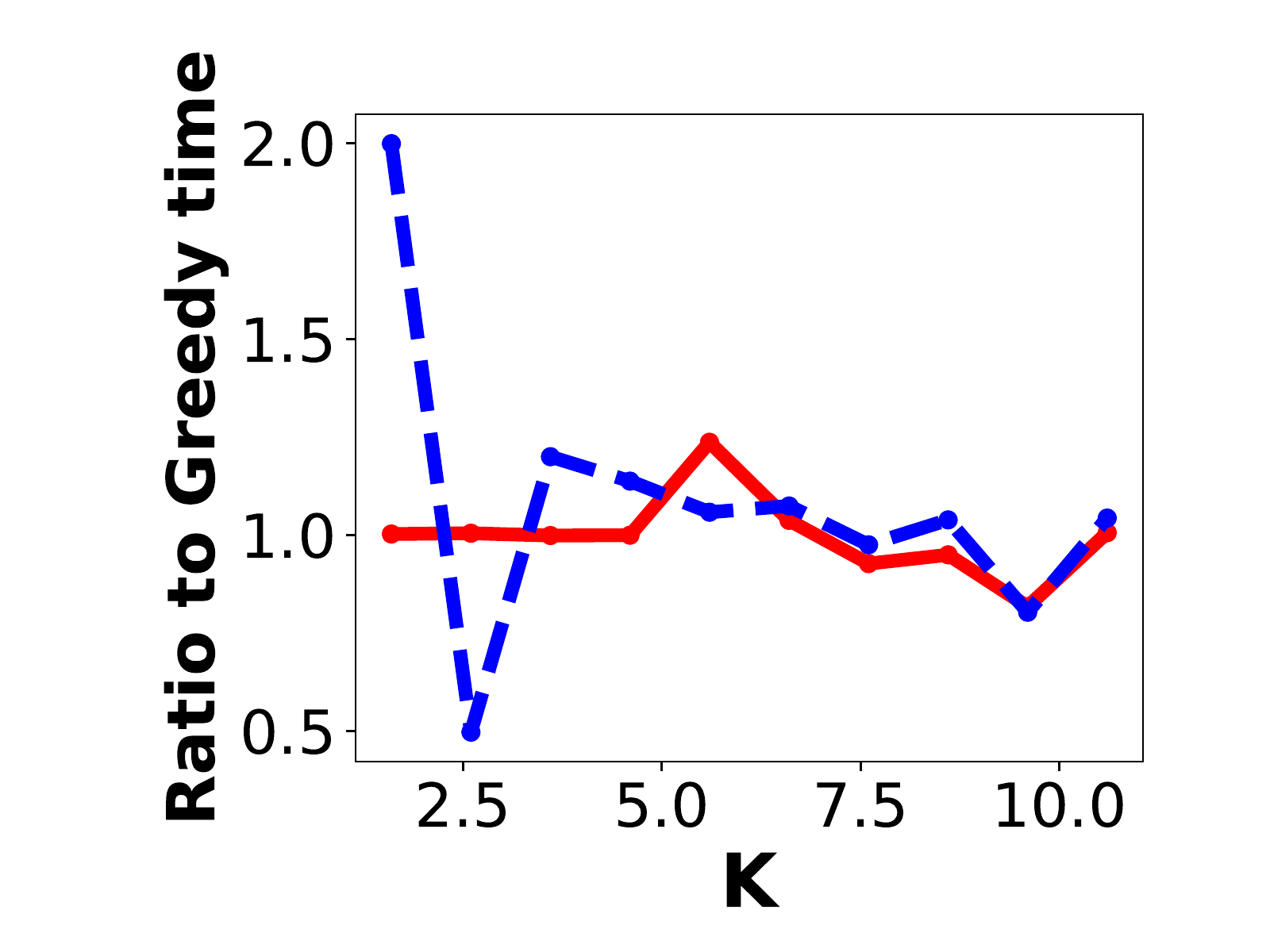}}
	\subfloat[\texttt{ml-20}]{\includegraphics[width=0.32\textwidth]{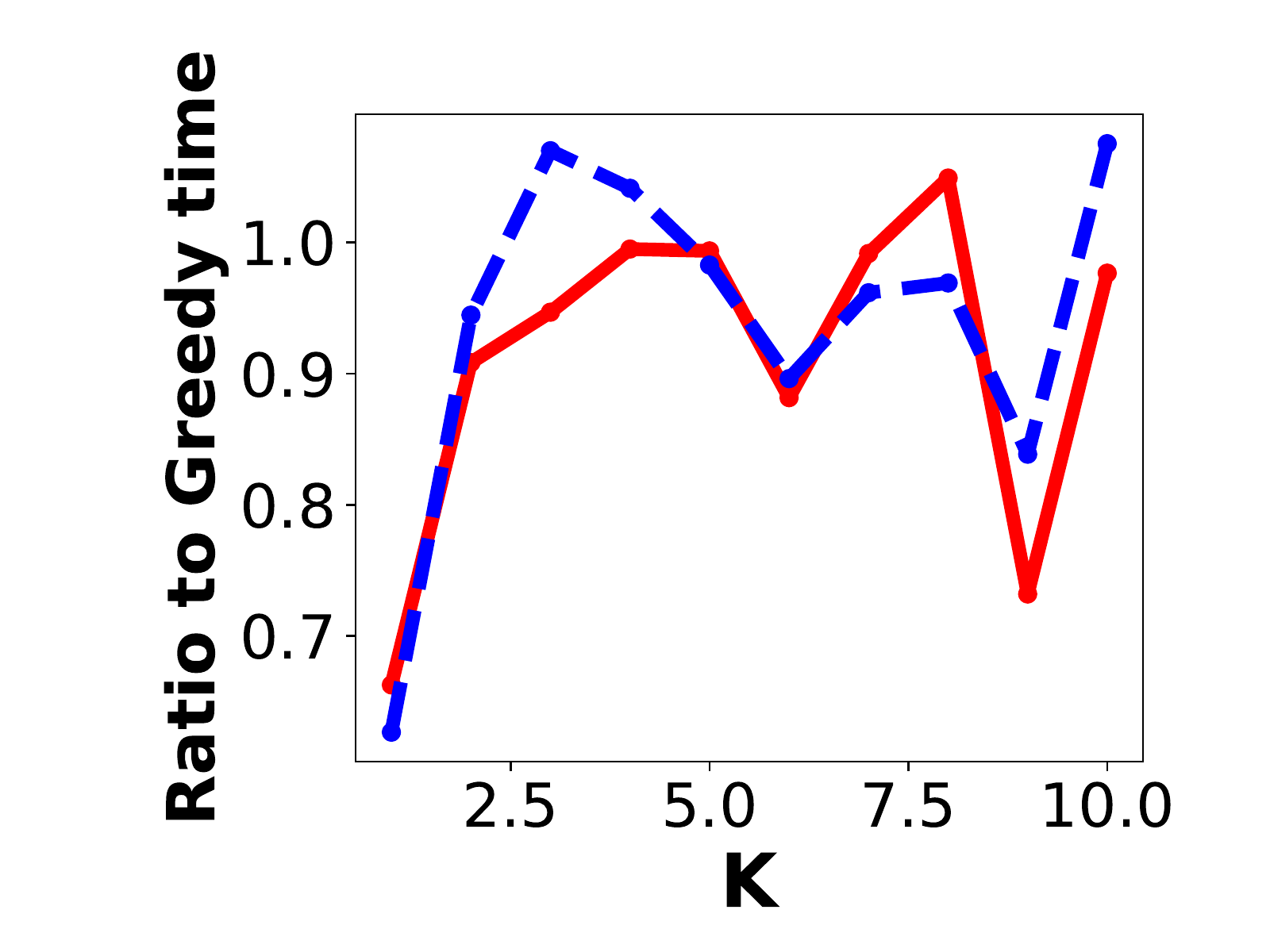}\label{fig:times_end}}\\
	\subfloat{\includegraphics[width=0.5\textwidth]{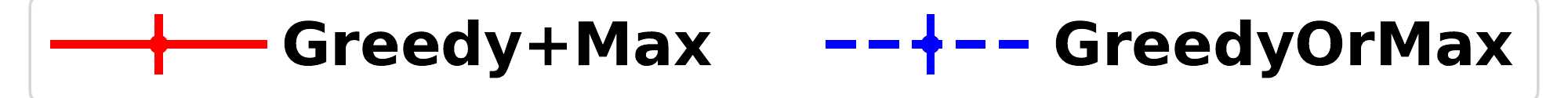}}
	\caption{Ratio of runtime of offline algorithms to the runtime of \greedy, for different values of $K$. Observe that \gmax and \gormax show similar running time, which is at most $20\%$ greater than \greedy running time. The ratio of \partenum  runtime is not displayed, due to it being several orders of magnitude larger, e.g., $1000$ times larger for $K=15$.}
	\label{fig:offline_times}
\end{figure*}
\begin{figure*}[!ht]
	\centering
	\subfloat[\texttt{com-dblp}]{\includegraphics[width=0.32\textwidth]{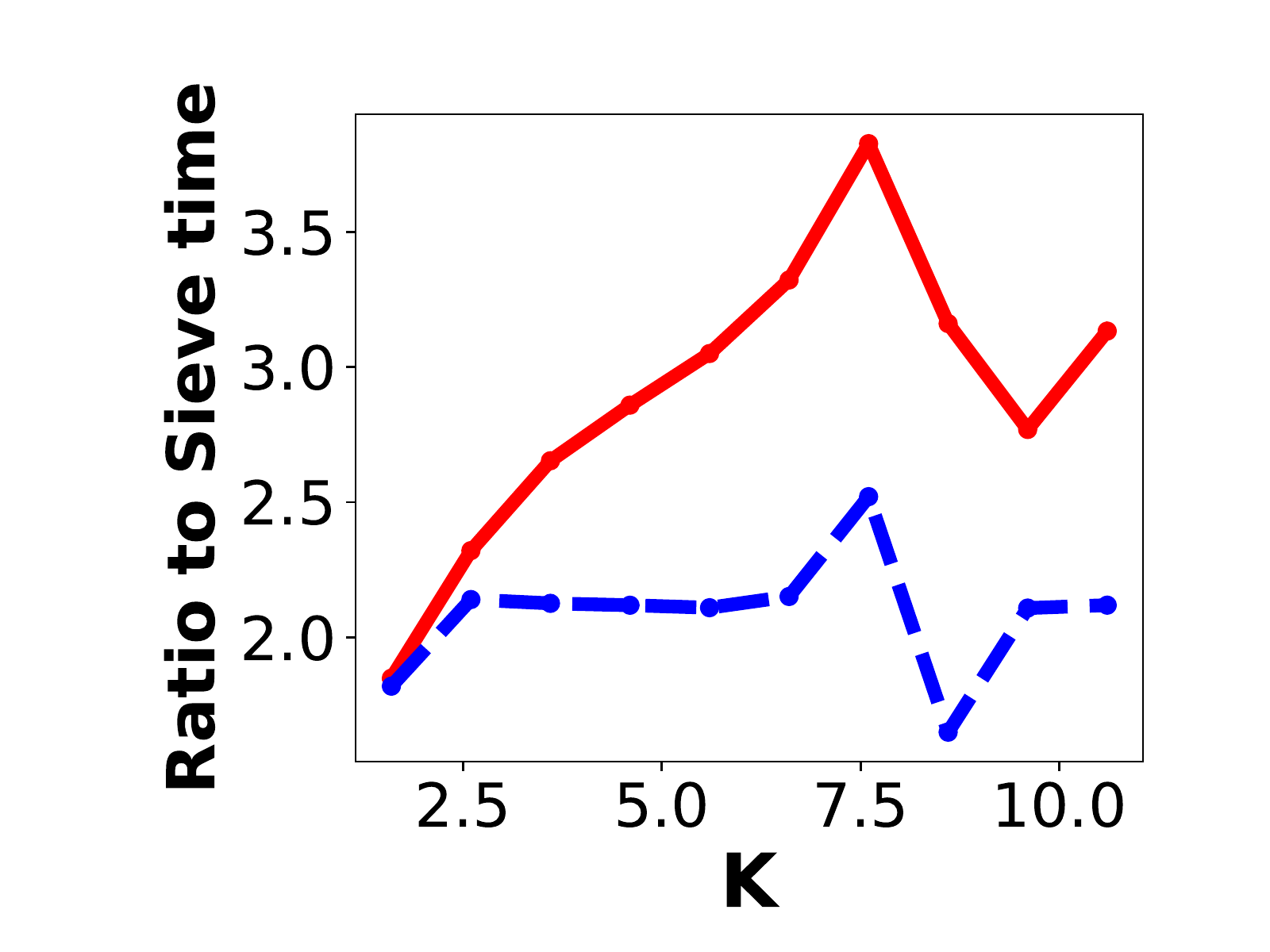}\label{fig:streaming_times_start}}
	\subfloat[\texttt{ego-Facebook}]{\includegraphics[width=0.32\textwidth]{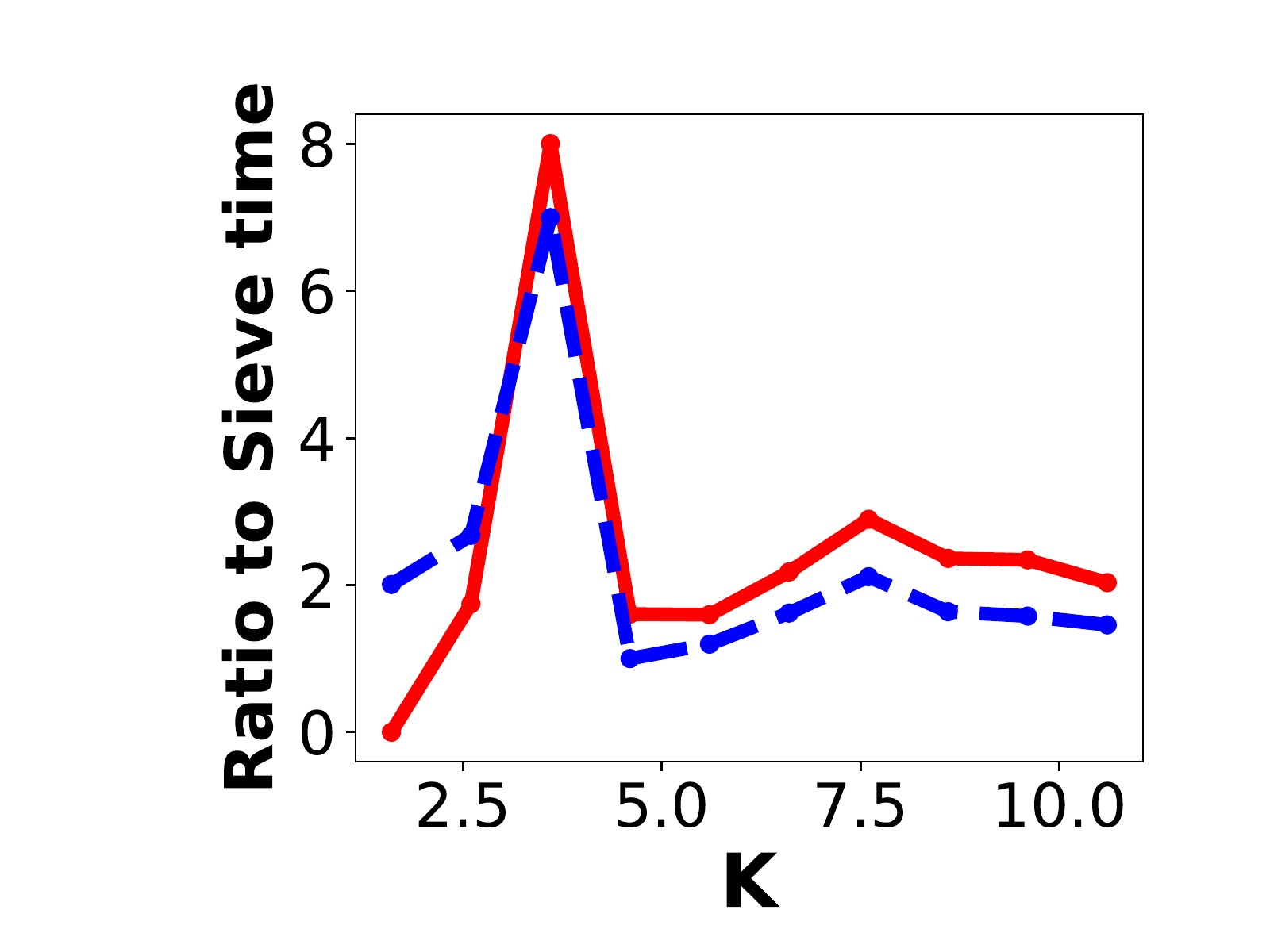}}
	\subfloat[\texttt{ml-20}]{\includegraphics[width=0.32\textwidth]{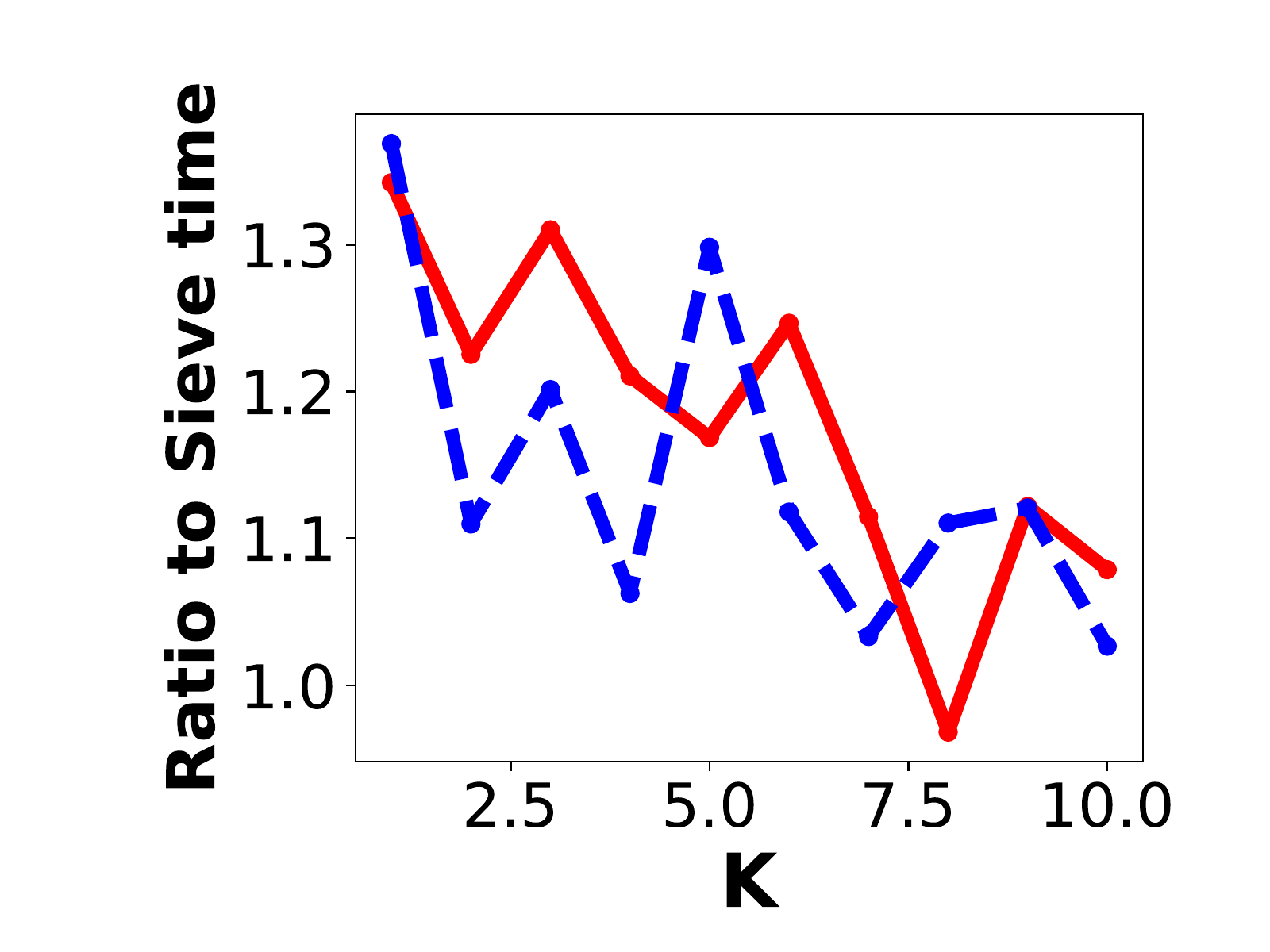}\label{fig:streaming_times_end}}\\
	\vspace{-0.3cm}
	\subfloat{\includegraphics[width=0.5\textwidth]{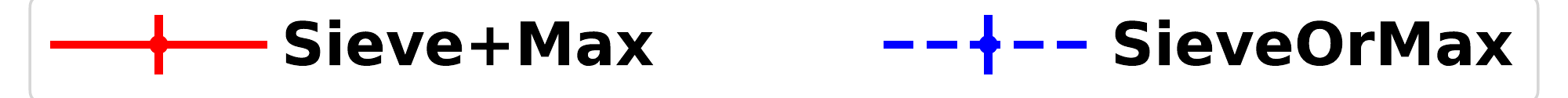}}
	\caption{Ratio of average runtime of streaming algorithms to the average runtime of \sieve for different values of $K$, across ten iterations. The larger ratios can be explained from the oracle calls made on larger sets by \smax being more expensive than the average oracle call made by \sieve. The ratio of \branching runtime is not displayed, due to being several orders of magnitude larger, e.g., 80K times larger for $K=5$.}
	\label{fig:streaming_times}
\end{figure*}
\begin{figure*}[!ht]
	\centering
	\subfloat[\texttt{com-dblp}]{\includegraphics[width=0.32\textwidth]{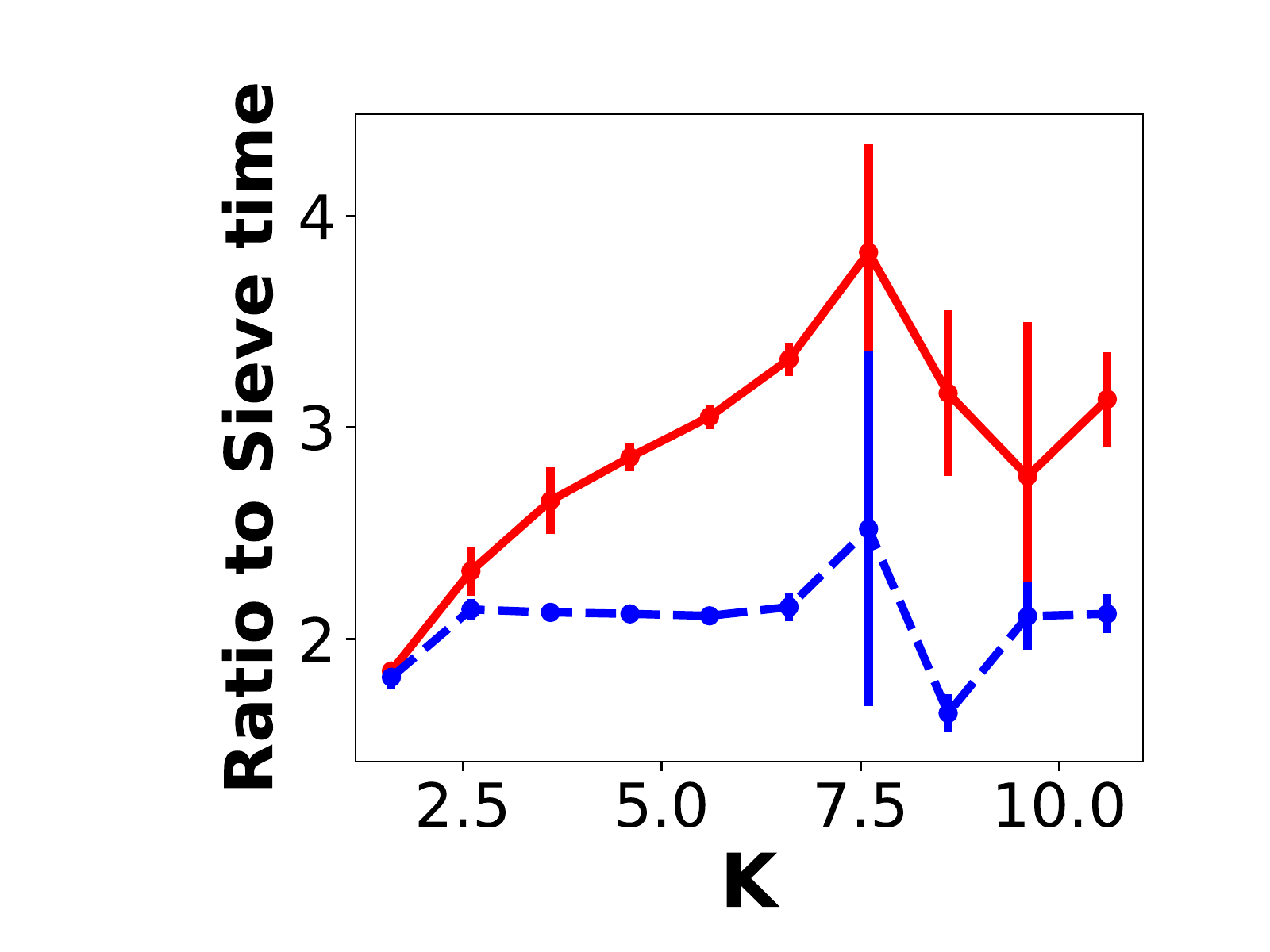}}
	\subfloat[\texttt{ego-Facebook}]{\includegraphics[width=0.32\textwidth]{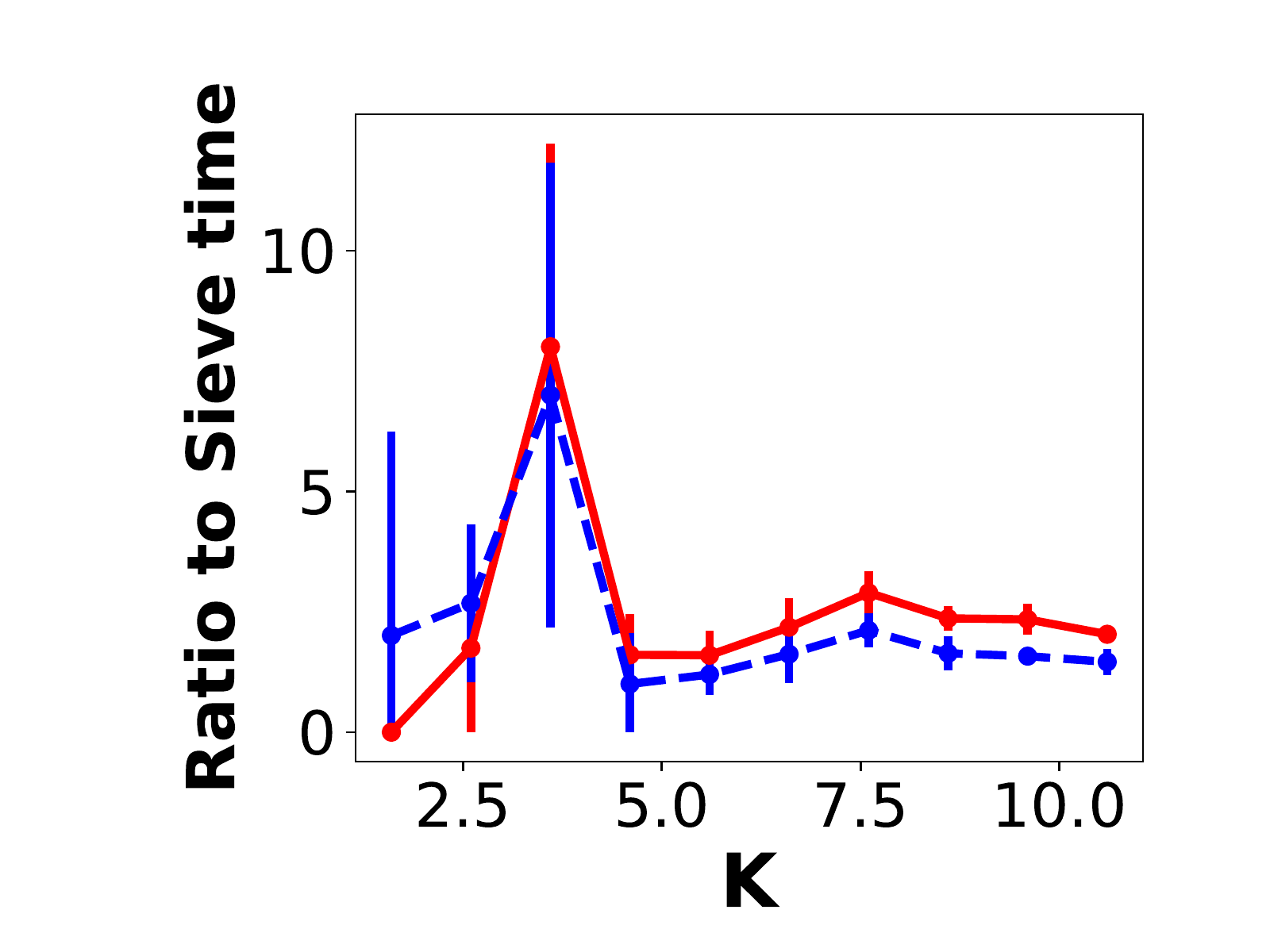}}
	\subfloat[\texttt{ml-20}]{\includegraphics[width=0.32\textwidth]{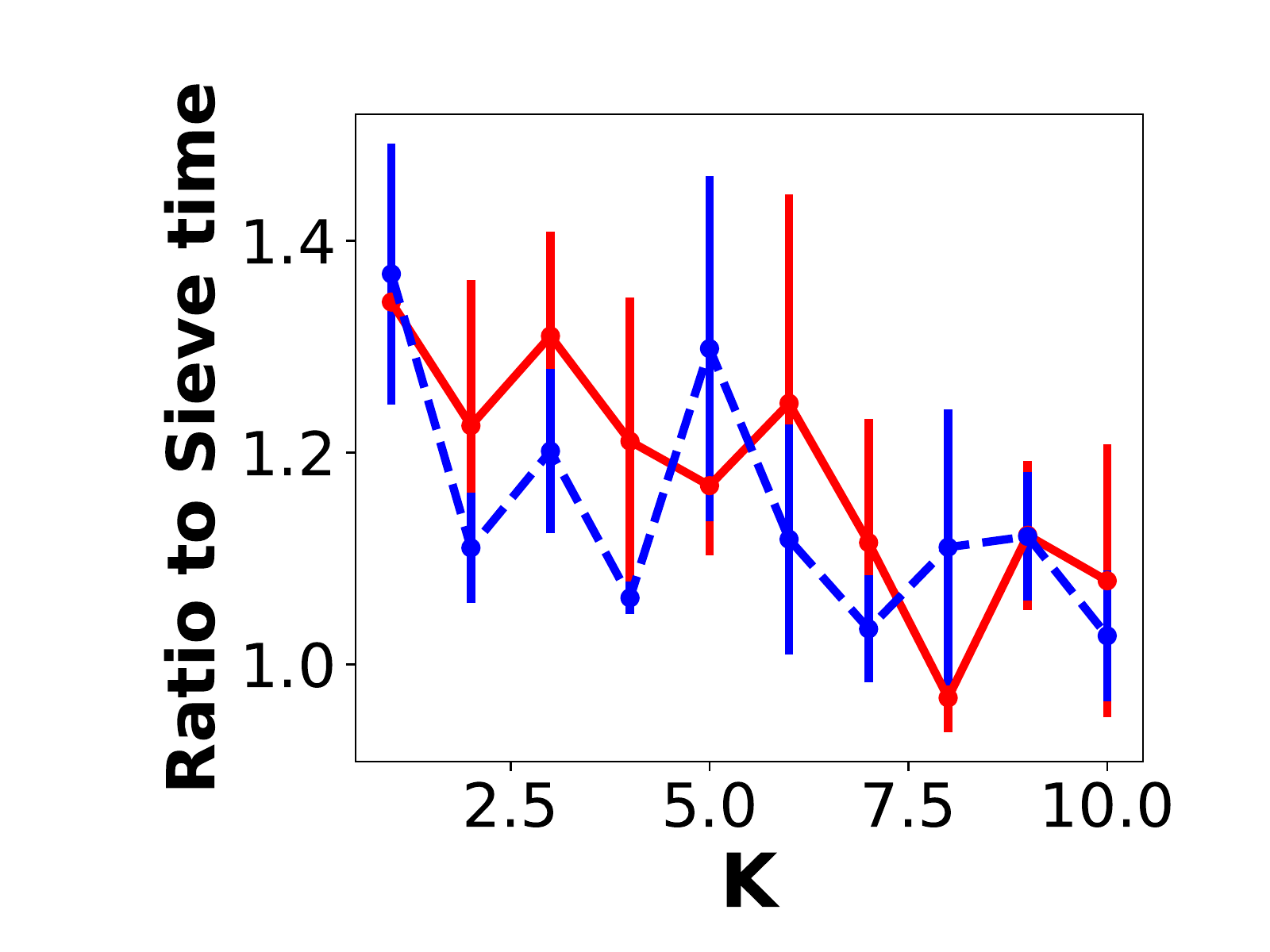}}\\
	\vspace{-0.3cm}
	\subfloat{\includegraphics[width=0.5\textwidth]{experiments/pics/legend_streaming_times.pdf}}
    \caption{Ratio of average runtime of streaming algorithms compared to the average runtime of \sieve, with error bars representing one standard deviation for each algorithm on the corresponding knapsack constraint across ten iterations.}
	\label{fig:streaming_times_bars}
\end{figure*}
\fi
\ifarxiv
\paragraph{Oracle calls.} 
We also compare the number of oracle calls performed by the algorithms. 
\gmax, \gormax and \greedy require the same amount of oracle calls, since computing marginal gains and finding the best element for augmentation compute the objective on the same set.
On the other hand, \partenum requires $544$x more calls than \greedy for $K=8$.
For the streaming algorithms, the number of oracle calls made by \sieve, \smax, and \sieve, never differed by more than a factor of two, while \branching requires a factor 125K more oracle calls than \sieve for $K=8$. 
We illustrate the number of oracle calls made by these algorithms in Figure~\ref{fig:streaming_oracles}. 
\begin{figure*}[!ht]
	\centering
	\subfloat[\texttt{com-dblp}]{\includegraphics[width=0.32\textwidth]{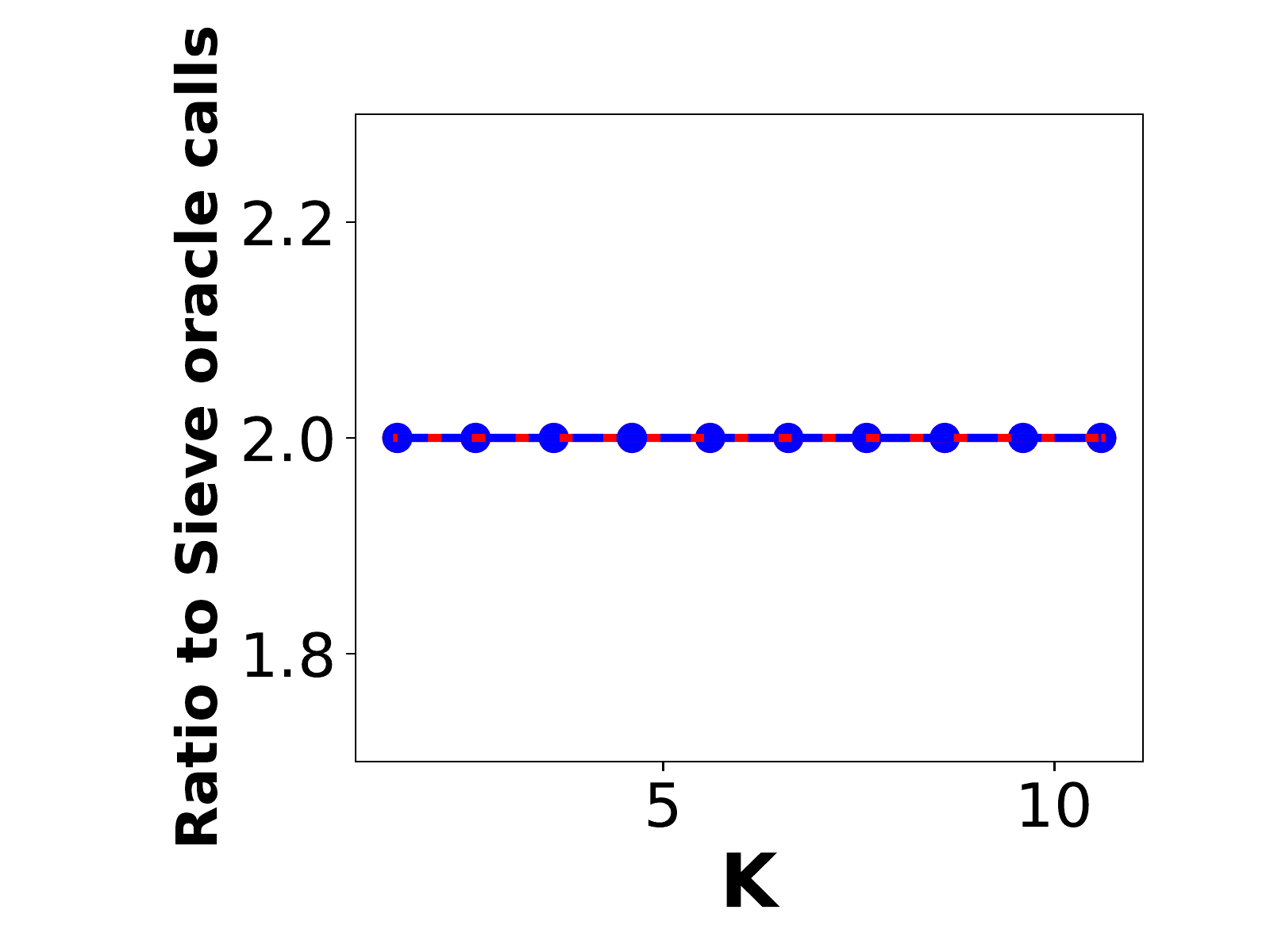}}
	\subfloat[\texttt{ego-Facebook}]{\includegraphics[width=0.32\textwidth]{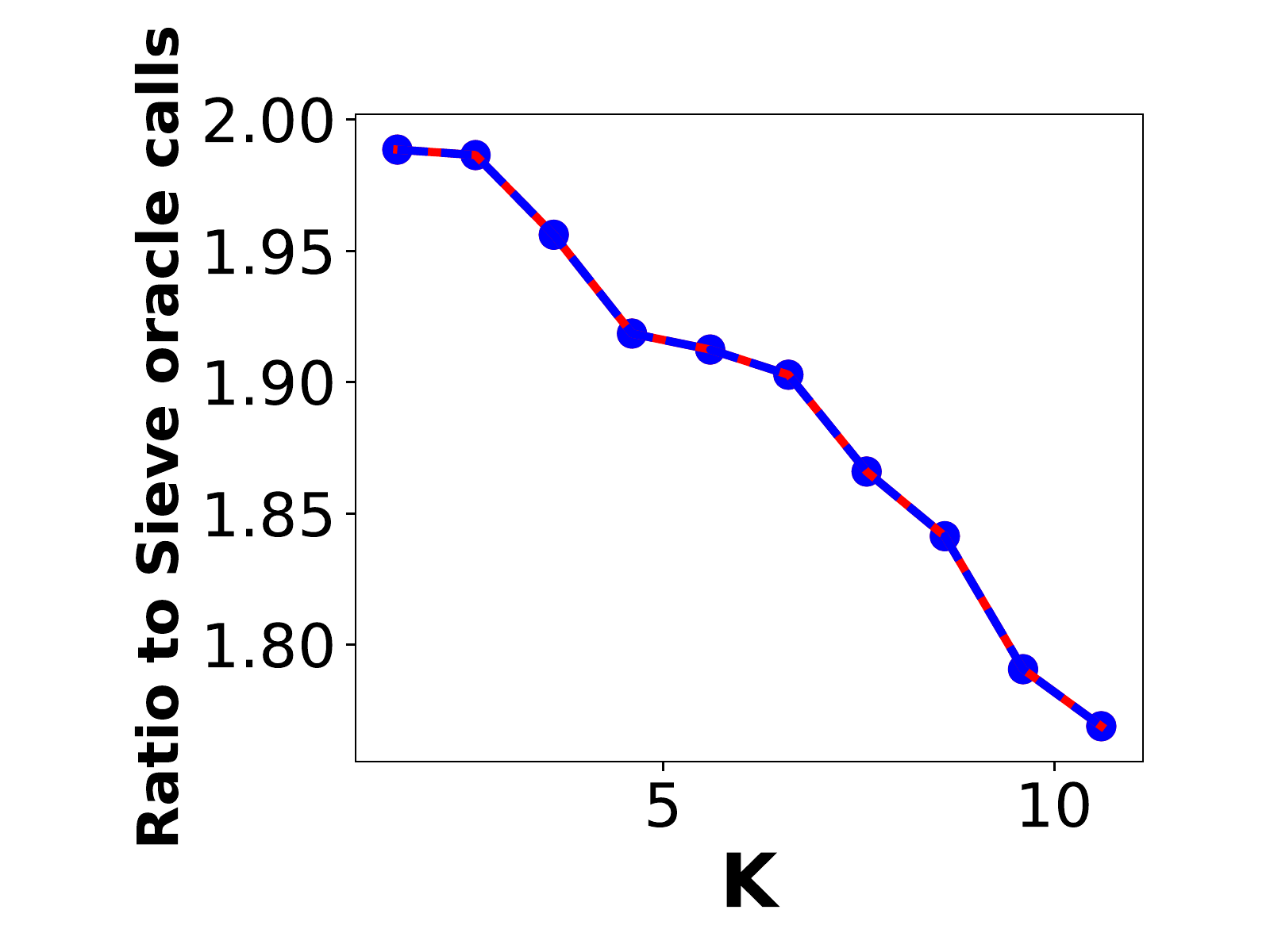}}
	\subfloat[\texttt{ml-20}]{\includegraphics[width=0.32\textwidth]{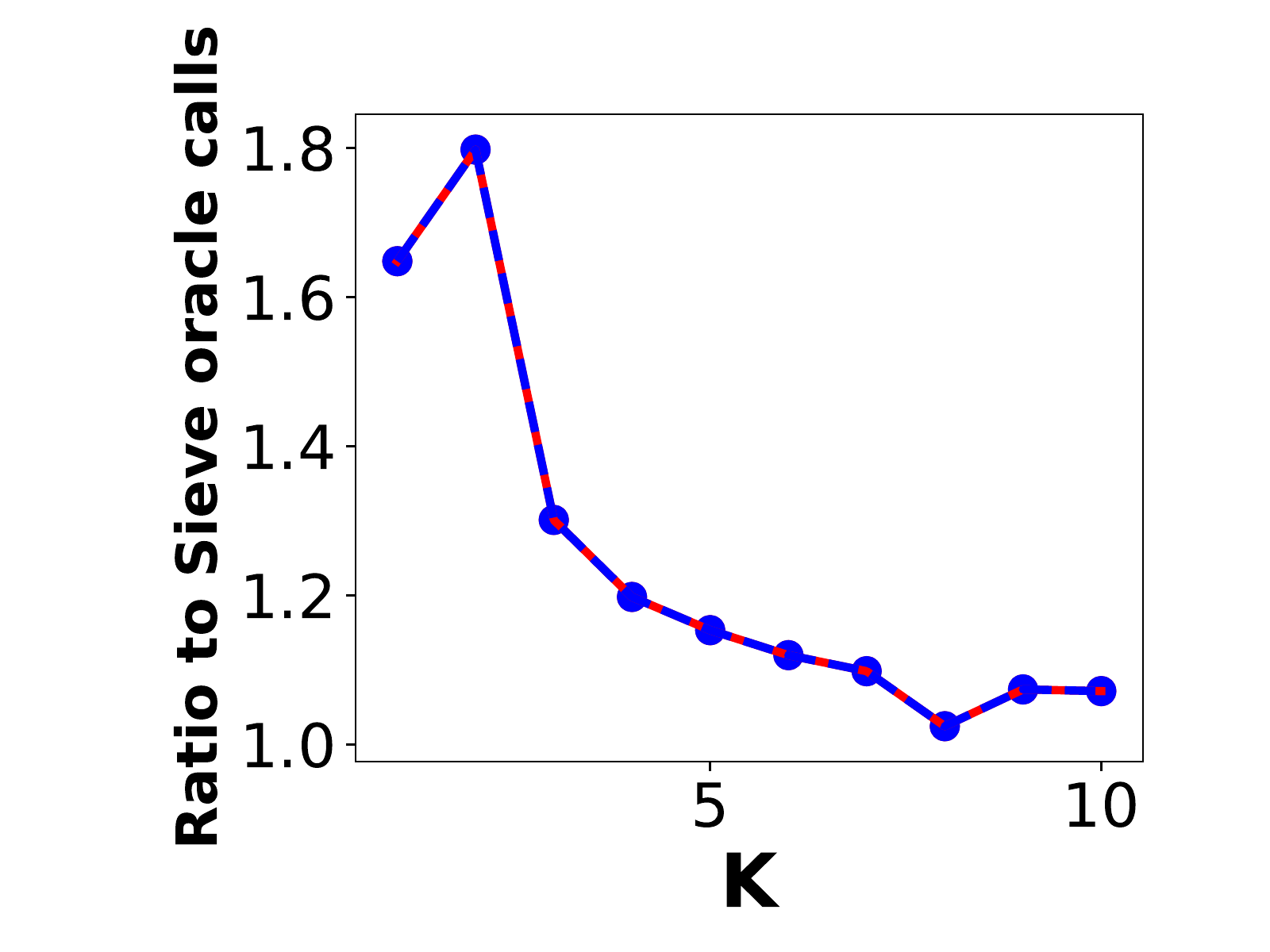}}\\
	\caption{Smoothed ratio of average number of oracle calls made by streaming algorithms compared to the average number of oracles calls made by \sieve, across ten iterations.}
	\label{fig:streaming_oracles}
\end{figure*}
\fi
\newpage
\bibliography{references}
\bibliographystyle{alpha}
\arxiv{
\appendix
\section{Standard greedy and thresholding inequalities}

In this section we prove the standard greedy inequality $\gval(x) \ge 1 - e^{-x}$, where $x$ is the cost of a partial greedy solution. To prove it, we first show that a differential inequality $\gval(x) + \gval'(x) \ge 1$ holds, and then integrate it using Proposition~\ref{prop:difEq}. For the thresholding algorithm a similar approximate inequality holds.

\begin{proposition}
	\label{prop:difEq}
	Let $\xi$ be a continuous and piecewise smooth function $[u, v] \to \Rp$.
	If for some $\alpha, \beta>0$ we have $\xi(x)+\alpha \xi'(x)\geq \beta$ for $u\leq x \leq v$, then
	$\xi(v)\geq  \beta+(\xi(u)-\beta)\exp{\frac{u-v}{\alpha}}$.
\end{proposition}

\begin{proof}
    First, consider the case when $\xi$ is smooth.
	$\xi(x)+\alpha \xi'(x)\geq \beta$ implies that $\xi(x)\exp{\frac{x}{\alpha}}+\alpha \xi'(x)\exp{\frac{x}{\alpha}}\geq \beta\exp{\frac{x}{\alpha}}$ through multiplication by $\exp{\frac{x}{\alpha}}$. 
	Observe that $\xi(x)\exp{\frac{x}{\alpha}}+\alpha \xi'(x)\exp{\frac{x}{\alpha}}$ is the derivative of $\xi(x)\alpha\exp{\frac{x}{\alpha}}$. 
	Hence, $\frac{d(\xi(x)\alpha\exp{\frac{x}{\alpha}})}{dx}\geq \beta\exp{\frac{x}{\alpha}}$ implies
	\begin{align*}
	\int_{u}^{v}d(\xi(x)\alpha\exp{\frac{x}{\alpha}})
	&\geq
	\int_{u}^{v} \beta\exp{\frac{x}{\alpha}}{dx}\\
	(\xi(x)\alpha\exp{\frac{x}{\alpha}}) \Big|_{u}^{v}
	&\geq \alpha\beta\exp{\frac{x}{\alpha}}\Big|_{u}^{v}\\
	\xi(v)\alpha\exp{\frac{v}{\alpha}} - \xi(u)\alpha\exp{\frac{u}{\alpha}} 
	&\geq
	\alpha\beta\exp{\frac{v}{\alpha}} - \alpha\beta\exp{\frac{u}{\alpha}}.
	\end{align*}
	Dividing both sides by $\alpha$, 
	\begin{align*}
	\xi(v)\exp{\frac{v}{\alpha}} - \xi(u)\exp{\frac{u}{\alpha}}
	&\geq \beta\exp{\frac{v}{\alpha}} - \beta\exp{\frac{u}{\alpha}} \\
	\xi(v)
	&\geq 
	\beta+(\xi(u)-\beta)\exp{\frac{u-v}{\alpha}}.
	\end{align*}
	
	For a piecewise smooth $\xi$, let $u = x_0 < x_1 < \cdots < x_t = v$, such that $\xi$ is smooth on a segment $(x_{i}, x_{i+1})$ for any $i$.
	By induction, we prove that the inequality holds for $x_0, x_i$ for any $i$:
	$$\xi(x_i) \geq  \beta+(\xi(x_0)-\beta)\exp{\frac{x_i-x_0}{\alpha}}.$$
	The statement is true for $i=0$. Induction step:
	\begin{align*}
	    \xi(x_{i+1})
	    &\geq \beta + (\xi(x_i)-\beta)\exp{\frac{x_{i}-x_{i+1}}{\alpha}}  \\
	    &\geq \beta + (\xi(x_0)-\beta)\exp{\frac{x_0-x_i}{\alpha}}\exp{\frac{x_i-x_{i+1}}{\alpha}}  \\
	    &\geq \beta + (\xi(x_0)-\beta)\exp{\frac{x_0-x_{i+1}}{\alpha}}
	\end{align*}
\end{proof}

\begin{theorem}[Standard greedy inequality]\label{lem:gEq}
For all $x \in [0, 1-c\left(\best\right)]$, the greedy performance function $\gval$ satisfies the following differential inequality: 
\[\gval(x)+\gval'(x)\geq 1,\]
and hence also its integral version: $\gval(x) \ge 1 - e^{-x}$.
\end{theorem}

\ifarxiv
\begin{proof}
Let $x\in[0, 1- c\left(\best\right)]$ and recall that by definition $\calG_{i - 1}$ is the largest set of elements selected by the greedy solution without exceeding total cost of $x$. 
Note that it suffices to show the inequality only for the left endpoints of the piecewise linear intervals of the form $[c(\calG_{i - 1}), c(\calG_i))$ as inside these intervals $\gval'$ stays constant while $\gval$ can only increase and hence the inequality holds.
Hence we can assume that $x = c(\calG_{i - 1})$ in the proof below which implies that $\gval(x) = f(\calG_{i - 1})$.

Since we normalized $f(\opt)=1$, by monotonicity: 
\begin{align*}
1&=f(\opt)\le f(\opt\cup \calG_{i - 1})\\
&= f(\calG_{i - 1})+\margain{\opt\setminus \calG_{i - 1}}{\calG_{i - 1}}.
\end{align*}
Then by submodularity and using the fact that by definition $f(\calG_{i - 1}) = g(x)$:
\begin{align*}
1 &\leq f(\calG_{i - 1})+\margain{\opt\setminus \calG_{i - 1}}{\calG_{i - 1}}\\
&\leq \gval(x) + \sum_{e\in \opt\setminus \calG_{i - 1}}\margain{e}{\calG_{i - 1}}.
\end{align*}
Since $\margain{e}{\calG_{i - 1}}=c(e)\marden{e}{\calG_{i - 1}}$:
\begin{align*}
1&\leq \gval(x) + \sum_{e\in \opt\setminus \calG_{i - 1}}\margain{e}{\calG_{i - 1}} \\
&= \gval(x) + \sum_{e\in \opt\setminus \calG_{i - 1}} c(e)\marden{e}{\calG_{i - 1}}\\ 
&\leq \gval(x) + \sum_{e\in \opt\setminus \calG_{i - 1}} c(e) \gval'(x),
\end{align*}
where the last inequality follows because greedy always picks the item with the largest marginal density and since $x \le 1 - c(\best)$ every item in $OPT \setminus \calG_{i - 1}$ can still fit into the knapsack. 
Hence,
\begin{align*}
1&\le  \gval(x) + \gval'(x)\sum_{e\in \opt\setminus \calG_{i - 1}} c(e)\\
&=  \gval(x) + \gval'(x) c(\opt\setminus \calG_{i - 1}).
\end{align*}
The desired differential inequality follows from the observation that $c(\opt\setminus \calG_{i - 1})\le c(\opt)\le 1$.
Finally, by integrating from $0$ to $x$ using the initial condition $\gval(0) = 0$, it follows that $\gval(x) \ge 1 - e^{-x}$ (by Proposition~\ref{prop:difEq}). 
\end{proof}
\fi

\begin{theorem}[Standard thresholding inequality]
\label{lem:APXgEq}
 For all $x \in [0 ,1-c(\best)]$, the thresholding performance function $\tval$ satisfies the following differential inequality: 
$$\tval(x)+(1+\epsilon)\tval'(x)\geq 1.$$
\arxiv{And hence also its integral version: $\tval(x) \ge 1 - e^{-\frac{x}{1+\eps}}$.}
\end{theorem}
\ifarxiv

\begin{proof}	
Let $p \in [0,1]$ be the total cost of the elements collected by the thresholding algorithm in the first pass.
First, note that for the first pass when $x \in [0,p]$ the differential inequality follows trivially as $\tval'(x) \ge \frac{\lambda}{\alpha K} \ge 1$ since $\lambda \ge \alpha f(\opt)$ and by our normalization $f(\opt) = K = 1$.
Fix $x\in[p, 1- c\left(\best\right)]$ and recall that by definition $\calT_{i - 1}$ is the largest set of elements selected by the thresholding algorithm without exceeding total cost of $x$. Similarly to the previous proofs it suffices to consider only the left endpoints of the intervals of the form $[c(\calT_{i - 1}, \calT_i))$ so we assume $x = c(\calT_{i - 1})$.
Since we normalized $f(\opt)=1$, then by monotonicity: 
\begin{align*}
1&=f(\opt) \le f(\opt\cup \calT_{i-1}) \\
&= f(\calT_{i-1})+\margain{\opt\setminus \calT_{i-1}}{\calT_{i-1}}.    
\end{align*}
Then by submodularity and using the fact that by definition $\tval(x) = f(\calT_{i - 1})$:
\begin{align*}
1 &\le f(\calT_{i-1})+\margain{\opt\setminus \calT_{i-1}}{\calT_{i-1}}\\
&\leq \tval(x) + \sum_{e\in \opt\setminus \calT_{i-1}}\margain{e}{\calT_{i-1}}.
\end{align*}
Since $\margain{e}{\calT_{i-1}}=c(e)\marden{e}{\calT_{i-1}}$:
\begin{align*}
1&\leq \tval(x) + \sum_{e\in \opt\setminus \calT_{i-1}}\margain{e}{\calT_{i-1}}\\
&= \tval(x) + \sum_{e\in \opt\setminus \calT_{i-1}} c(e)\marden{e}{\calT_{i-1}}\\
&\le \tval(x) + \sum_{e\in \opt\setminus \calT_{i-1}} c(e) \tval'(x)(1+\eps),
\end{align*}
where the last inequality follows because after the first pass $\tval'(x) \ge \frac{\marden{e}{\calT_{i - 1}}}{1 + \eps}$ for all $e \in \opt \setminus \calT_{i - 1}$. Indeed, note that in all passes except the first one the thresholding algorithm always selects an item whose marginal density is at least $(1 + \eps)^{-1}$ times the best marginal density available. Since $t'(x)$ is the density of this item and all items in $\opt \setminus \calT_{i - 1}$ still fit (as $x \le 1 - c(\best)$) we have $(1 + \eps)\tval'(x) \ge \max_{e \in \opt \setminus \calT_{i - 1}} \marden{e}{\calT_{i - 1}}$ as desired.
Hence:
\begin{align*}
1&\le \tval(x) + \sum_{e\in \opt\setminus \calT_{i-1}} c(e) \tval'(x)(1+\eps)\\
&=  \tval(x) + (1+\eps)\tval'(x) c(\opt\setminus \calT_{i-1}).
\end{align*}
The desired differential inequality follows from the observation that $c(\opt\setminus \calT_{i-1})\le c(\opt)=1$.

For the integral version we integrate the differential inequality between $0$ and $x$ with the initial condition $t(0) =  0$ (formally, apply Proposition~\ref{prop:difEq} with $\alpha = 1+\eps, \beta = 1, u = 0, v = x$) and get $t(x) \ge 1 - e^{-\frac{x}{1+\eps}}$, as desired.
\end{proof}
\fi

\section{Omitted proofs}
\begin{fact}
\label{fact:calculation}
For all $0\le x\le 1$,
\[(1-x)e^{2x-1}\le\frac{1}{2}.\]
\end{fact}
\begin{proof}
Let $r(x)=(1-x)e^{2x-1}$ and note that $r'(x)=(1-2x)e^{2x-1}$ so that $r'(x)>0$ for $x\in\left[0,\frac{1}{2}\right]$ and $r'(x)\le0$ for $x\in\left[\frac{1}{2},1\right]$. 
Hence, it follows that $r\left(\frac{1}{2}\right)=\frac{1}{2}$ is a local maximum and so $(1-x)e^{2x-1}\le\frac{1}{2}$ for all $0\le x\le 1$. 
\end{proof}

\begin{fact}
\label{fact:second:calculation}
\[\left(1-\frac{c\left(o_1\right)}{1+\eps}\right)\exp {\frac{2 c\left(o_1\right) - 1}{1+\eps}}\le\frac{1}{2}+\eps.\]
\end{fact}
\begin{proof}
By Fact~\ref{fact:calculation},
\[\left(1-\frac{c\left(o_1\right)}{1+\eps}\right)\exp {\frac{2 c\left(o_1\right) - 1}{1+\eps}}\le\frac{1}{2}\exp {\frac{\eps}{1+\eps}}.\]
Hence it suffices to show that $\exp {\frac{\eps}{1+\eps}}\le 1+2\eps$, which follows from the fact that $\frac{d}{dx}\exp{\frac{x}{1+x}}\le 2$ for $0\le x\le 1$.
\end{proof}
We now describe a generalization to a knapsack constraint of the algorithm of \cite{KazemiMZLK19} that computes a constant factor approximation to maximum submodular maximization under a cardinality constraint, using small space and a small number of queries. 
\begin{algorithm}
\textbf{Input}: Stream of elements $E=e_1,\ldots, e_n$, knapsack capacity $K$, cost function $c(\cdot)$, non-negative monotone submodular function $f$, and an approximation parameter $\eps>0$\;
\textbf{Output}: A set $S$ that is a $\left(\frac{1}{3}-\epsilon\right)$-approximation for submodular maximization with a knapsack constraint\;
$\tau_{\min},\Delta,\LB\gets 0$\;
\For{each item $e_i$}{
    \If{$f(e_i)>\Delta$}{
        $e\gets e_i,\Delta\gets f(e_i)$
    }
    $\tau_{\min}=\frac{\max(2\LB,2\Delta)}{3K}$\;
    Discard all sets with $S_{\tau}$ with $\tau<\tau_{\min}$\;
    \For{$\tau\in\{(1+\eps)^i|\tau_{\min}/(1+\eps)\le(1+\eps)^i\le\Delta\}$}
    {
        \If{$\tau$ is a new threshold}{
            $S_{\tau}\gets\emptyset$
        }
        \If{$c(S_{\tau})<K$ and $\marden{e}{S_{\tau}}\ge\tau$}{
            $S_{\tau}\gets S_{\tau}\cup\{e\}$ and $\LB\gets\max\{\LB,f(S_{\tau})\}$
        }
    }
}
\Return $\argmax\{f(S_{\tau}),f(e)\}$
\caption{Space efficient constant factor approximation}\label{alg:constantStream}
\end{algorithm}

\begin{theorem}
\label{thm:constant:stream}
There exists a one-pass streaming algorithm that outputs a $\left(\frac{1}{3}-\epsilon\right)$-approximation to the submodular maximization under knapsack constraint that uses $O\left(\frac{K}{\eps}\right)$ space and $O\left(\frac{n\log K}{\eps}\right)$ total queries.
\end{theorem}
\begin{proof}
Since Algorithm~\ref{alg:constantStream} uses the same threshold as Algorithm 2 in~\cite{HuangKY17}, it outputs a $\frac{1}{3}-\epsilon$-approximation. 
On the other hand, by Theorem 1 in~\cite{KazemiMZLK19}, Algorithm~\ref{alg:constantStream} uses space $O\left(\frac{K}{\eps}\right)$ and query complexity $O\left(\frac{n\log K}{\eps}\right)$. 
\end{proof}
Hence, by setting $\eps=\frac{1}{6}$, we obtain the following:
\begin{corollary}
There exists a one-pass streaming algorithm that outputs a $\frac{1}{6}$-approximation to the submodular maximization under knapsack constraint that uses $O(K)$ space and $O(n\log K)$ total queries.
\end{corollary}

\ignore{
\subsection{Basic inequalities}
The case when $u=f(u)=0$ in Proposition~\ref{prop:difEq} is of particular interest to us.
\begin{corollary}
\label{cor:ex}
For any $\alpha, \beta>0$, $v\in [0,1]$ and any function $f:[0,1]\rightarrow \Rp$ such that $f(x)+\alpha f'(x)\geq \beta$ for $0\leq x \leq v$, we have $f(v)\geq\beta(1-\exp{-\frac{v}{\alpha}})$.
\end{corollary}
}
}

\end{document}